\documentclass[conference]{IEEEtran}
\IEEEoverridecommandlockouts 						
\usepackage{amsmath,amssymb}
\usepackage{amsthm}
\usepackage{color}
\usepackage[mathscr]{eucal}
\usepackage{graphics,graphicx,multicol}
\usepackage{epsfig}
\usepackage{enumerate}
\usepackage{subfigure}
\usepackage{algorithmic}
\usepackage{algorithm}
\usepackage{morefloats}
\usepackage{enumerate}
\usepackage{subfigure}
\usepackage{multirow}
\usepackage{xfrac}
\usepackage{array}
\usepackage{pstricks, pst-node, pst-plot, pst-circ}
\usepackage{pst-3dplot,pstricks-add,pst-solides3d,pst-grad}
\usepackage{moredefs}
\usepackage{cite}
\usepackage{courier}
\newtheorem{thm}{Theorem}[section] 
\newtheorem{lem}[thm]{Lemma}

\begin{document}

\title{An Optimal Resource Allocation with Joint Carrier Aggregation in 4G-LTE}
\author{Ahmed Abdelhadi and T. Charles Clancy \\
Hume Center, Virginia Tech, Arlington, VA, 22203, USA\\
\{aabdelhadi, tcc\}@vt.edu
}
\maketitle

\begin{abstract}
In this paper, we introduce a novel approach for optimal resource allocation from multiple carriers for users with elastic and inelastic traffic in fourth generation long term evolution (4G-LTE) system. In our model, we use logarithmic and sigmoidal-like utility functions to represent the user applications running on different user equipments (UE)s. We use utility proportional fairness policy, where the fairness among users is in utility percentage of the application running on the mobile station. Our objective is to allocate the resources to the users optimally from multiple carriers. In addition, every user subscribing for the mobile service is guaranteed to have a minimum quality-of-service (QoS) with a priority criterion. Our rate allocation algorithm selects the carrier or multiple carriers that provide the minimum price for the needed resources. We prove that the novel resource allocation optimization problem with joint carrier aggregation is convex and therefore the optimal solution 
is tractable. We present a distributed algorithm to allocate the resources optimally from multiple evolved NodeBs (eNodeB)s. Finally, we present simulation results for the performance of our rate allocation algorithm.
\end{abstract}

\begin{keywords}
Optimal Resource Allocation, Joint Carrier Aggregation, Inelastic Traffic
\end{keywords}
\pagenumbering{gobble}

\providelength{\AxesLineWidth}       \setlength{\AxesLineWidth}{0.5pt}%
\providelength{\plotwidth}           \setlength{\plotwidth}{8cm}
\providelength{\LineWidth}           \setlength{\LineWidth}{0.7pt}%
\providelength{\LineWidthTwo}        \setlength{\LineWidthTwo}{1.7pt}%
\providelength{\MarkerSize}          \setlength{\MarkerSize}{3.5pt}%
\newrgbcolor{GridColor}{0.8 0.8 0.8}%
\newrgbcolor{GridColor2}{0.5 0.5 0.5}%

\section{Introduction}\label{sec:intro}

In recent years, mobile broadband systems have witnessed rapid growth in both the number of subscribers and the traffic of each subscriber. Mobile subscribers are currently running multiple applications, simultaneously, on their smart phones. The network providers are moving from single service (e.g. Internet access) to multiple service offering (e.g. multimedia telephony and mobile-TV) \cite{QoS_3GPP}. In order to meet this strong demand for wireless resources by the mobile users more resources are needed \cite{Carrier_Agg_1}. However, due to the scarcity of the spectrum, it is difficult to have a single frequency band fulfilling this demand. Therefore, resources from different carriers need to be aggregated, leading to interband non-contiguous carrier aggregation \cite{Carrier_Agg_2}.

In addition, The National Broadband Plan (NBP) and the findings of the President's Council of Advisors on Science and Technology (PCAST) spectrum study have recommended that under-utilized federal spectrum be made available for commercial use \cite{PCAST}. Making more spectrum available will certainly provide opportunities for mobile broadband capacity gains, but only if those resources can be aggregated efficiently with the existing commercial mobile system. The efficient  non-contiguous carrier aggregation of federal spectrum into the existing cellular network is a challenging task. The challenges are both in hardware implementation and joint resource allocation. Hardware implementation challenges are in the need for multiple oscillators, multiple RF chains, more powerful signal processing, and longer battery life \cite{RebeccaThesis}. Regarding resource allocation, a distributed resource allocation algorithm is needed to optimally allocate resources from different carriers of different network providers.

In this paper, we focus on joint resource allocation from multiple carriers. We formulate the resource allocation optimization problem with joint carrier aggregation into a convex optimization framework. We use logarithmic and sigmoidal-like utility functions to represent delay-tolerant and real-time applications, respectively \cite{Ahmed_Utility1}. Our model supports both contiguous and non-contiguous carrier aggregation from one or more network providers. In the rate allocation process, our distributed algorithm allocates resources from one or more carriers to provide the lowest resource prices for the mobile users. In addition, our algorithm uses utility proportional fairness policy to give priority to real-time applications over delay-tolerant applications when allocating resources.

\subsection{Related Work}\label{sec:related}

In \cite{Ahmed_Utility1, Ahmed_Utility2, Ahmed_Utility3}, the authors present an optimal rate allocation algorithm for users connected to a single carrier. The optimal rates are achieved by formulating the rate allocation optimization problem in a convex optimization framework. The authors use logarithmic and sigmoidal-like utility functions to represent delay-tolerant and real-time applications, respectively. In \cite{Ahmed_Utility1}, the rate allocation algorithm gives priority to real-time applications over delay-tolerant applications when allocating resources as the utility proportional fairness rate allocation policy is used. 

In \cite{Haya_Utility1}, the authors present multiple-stage carrier aggregation with optimal resource allocation algorithm with utility proportional fairness. The users allocate the resources from the first (primary) carrier eNodeB until all the resources in the eNodeB are allocated. The users switch to the second (secondary) carrier eNodeB to allocate more resources, and so forth. In \cite{Haya_Utility2}, spectrum sharing of public safety and commercial LTE bands is assumed. The authors presented a resource allocation algorithm with priority given to public safety users. The resource allocation algorithms in \cite{Haya_Utility1,Haya_Utility2} does not ensure optimal pricing where the allocation is performed in multiple stages. In this paper, we present an algorithm that allocates the resources jointly from different carriers and therefore ensures optimal rate allocation and optimal pricing.  

\subsection{Our Contributions}\label{sec:contributions}
Our contributions in this paper are summarized as:
\begin{itemize}
\item We introduce a novel rate allocation optimization problem with joint carrier aggregation that solves for utility functions that are logarithmic and sigmoidal-like representing delay-tolerant and real-time applications, respectively. 
\item In addition, we prove that the proposed optimization problem is convex and therefore the global optimal solution is tractable. \item We present a distributed rate allocation algorithm that converges to the optimal rates that maximize the optimization problem joint utility objective function. 
\item Our algorithm outperforms that presented in \cite{Haya_Utility1,Haya_Utility2}. It guarantees that mobile users receive minimum (optimal) price for resources. 
\end{itemize}

The remainder of this paper is organized as follows. Section \ref{sec:Problem_formulation} presents the problem formulation. Section \ref{sec:Proof} proves the global optimal solution exists and is tractable. In Section \ref{sec:Algorithm}, we present our distributed rate allocation algorithm with joint carrier aggregation for the utility proportional fairness optimization problem. Section \ref{sec:sim} discusses simulation setup and provides quantitative results along with discussion. Section \ref{sec:conclude} concludes the paper.

\section{Problem Formulation}\label{sec:Problem_formulation}

We consider 4G-LTE mobile system consisting of $K$ carriers eNodeBs with $K$ cells and $M$ UEs distributed in these cells. The rate allocated by the $l^{th}$ carrier eNodeB to $i^{th}$ UE is given by $r_{li}$ where $l =\{1,2, ..., K\}$ and $i = \{1,2, ...,M\}$. Each UE has its own utility function $U_i(r_{1i}+r_{2i}+ ...+r_{Ki})$ that corresponds to the type of traffic being handled by the $i^{th}$ UE. The utility function represents the percentage of user satisfaction eith the allocated rate. As UEs run different applications their utility functions are different. Our objective is to determine the optimal rates that the $l^{th}$ carrier eNodeB should allocate to the nearby UEs. We assume the utility functions $U_i(r_{1i}+r_{2i}+ ...+r_{Ki})$ to be a strictly concave or a sigmoidal-like functions. The strictly concave function is a good approximation for delay-tolerant applications,e.g. File Transfer Protocol (FTP), emails, etc. The S-shaped or sigmoidal-like utility function represents real-time 
applications, e.g. video streaming, Voice over IP (VoIP), etc. The utility functions have the following properties: 

\begin{itemize}
\item $U_i(0) = 0$ and $U_i(r_{1i}+r_{2i}+ ...+r_{Ki})$ is an increasing function of $r_{li}$ for $l$.
\item $U_i(r_{1i}+r_{2i}+ ...+r_{Ki})$ is twice continuously differentiable in $r_{li}$ for all $l$.
\end{itemize}
In our model, we use a normalized multi-variable sigmoidal-like utility function that can be expressed as 
\begin{equation}\label{eqn:sigmoid}
U_i(r_{1i}+r_{2i}+ ...+r_{Ki}) = c_i\Big(\frac{1}{1+e^{-a_i(\sum_{l=1}^{K}r_{li}-b_i)}}-d_i\Big)
\end{equation}
where $c_i = \frac{1+e^{a_ib_i}}{e^{a_ib_i}}$ and $d_i = \frac{1}{1+e^{a_ib_i}}$. So, it satisfies $U_i(0)=0$ and $U_i(\infty)=1$. See Figure \ref{fig:sim:Utilities} for a two dimensional view of sigmoidal-like utility of sigmoidal-like utility function $U_i(r_{1i} +r_{2i})$. We use a multi-variable normalized logarithmic utility function, as in \cite{UtilityFairness}, that can be expressed as 
\begin{equation}\label{eqn:log}
U_i(r_{1i}+r_{2i}+ ...+r_{Ki}) = \frac{\log(1+k_i\sum_{l=1}^{K}r_{li})}{\log(1+k_ir_{max})}
\end{equation}
where $r_{max}$ is the required rate for the user to achieve 100\% utilization and $k_i$ is the rate of increase of utilization with allocated rates. So, it satisfies $U_i(0)=0$ and $U_i(r_{max})=1$. See Figure \ref{fig:sim:Utilities} for a two dimensional view of sigmoidal-like utility of logarithmic utility function $U_i(r_{1i} +r_{2i})$. We consider the utility proportional fairness objective function given by 
\begin{equation}\label{eqn:utility_fairness}
\underset{\textbf{r}}{\text{max}} \prod_{i=1}^{M}U_i(r_{1i} + r_{2i} + ... + r_{Ki}) 
\end{equation}
where $\textbf{r} =\{\textbf{r}_1, \textbf{r}_2,..., \textbf{r}_M\}$ and $\textbf{r}_i =\{r_{1i}, r_{2i},..., r_{Ki}\}$. The goal of this resource allocation objective function is to allocate the resource for each UE that maximizes the total system utility while ensuring proportional fairness between utilities (i.e., the product of the utilities of all UEs). This resource allocation objective function ensures non-zero resource allocation for all users. Therefore, the corresponding resource allocation optimization problem provides a minimum QoS for all user. In addition, this approach allocates more resources to users with real-time applications which improves QoS for 4G-LTE system. 

The basic formulation of the utility proportional fairness resource allocation problem is given by the following optimization problem:
\begin{equation}\label{eqn:opt_prob_fairness}
\begin{aligned}
& \underset{\textbf{r}}{\text{max}} & & \prod_{i=1}^{M}U_i(r_{1i} + r_{2i} + ... + r_{Ki}) \\
& \text{subject to} & & \sum_{i=1}^{M}r_{1i} \leq R_1, ..., \:\:\sum_{i=1}^{M}r_{Ki} \leq R_K,\\
& & & r_{li} \geq 0, \;\;\;\;\;l = 1,2, ...,K, i = 1,2, ...,M.
\end{aligned}
\end{equation}
where $R_l$ is the total available rate at the $l^{th}$ carrier eNodeB.

We prove in Section \ref{sec:Proof} that the solution of the optimization problem (\ref{eqn:opt_prob_fairness}) is the global optimal solution. 
\section{The Global Optimal Solution}\label{sec:Proof}

In the optimization problem (\ref{eqn:opt_prob_fairness}), since the objective function $\arg \underset{\textbf{r}} \max \prod_{i=1}^{M}U_i(r_{1i}+r_{2i}+ ...+r_{Ki})$ is equivalent to $\arg \underset{\textbf{r}} \max \sum_{i=1}^{M}\log(U_i(r_{1i}+r_{2i}+ ...+r_{Ki}))$, then optimization problem (\ref{eqn:opt_prob_fairness}) can be written as:

\begin{equation}\label{eqn:opt_prob_fairness_mod}
\begin{aligned}
& \underset{\textbf{r}}{\text{max}} & & \sum_{i=1}^{M}\log \Big(U_i(r_{1i} + r_{2i} + ... + r_{Ki})\Big) \\
& \text{subject to} & & \sum_{i=1}^{M}r_{1i} \leq R_1,  ..., \:\:\sum_{i=1}^{M}r_{Ki} \leq R_K,\\
& & & r_{li} \geq 0, \;\;\;\;\;l = 1,2, ...,K, i = 1,2, ...,M.
\end{aligned}
\end{equation}

\begin{lem}\label{lem:concavity}
The utility functions $\log(U_i(r_{1i} + ... + r_{Ki}))$ in the optimization problem (\ref{eqn:opt_prob_fairness_mod}) are strictly concave functions. 
\end{lem}
\begin{proof}
In Section \ref{sec:Problem_formulation}, we assume that all the utility functions of the UEs are strictly concave or sigmoidal-like functions. 

In the strictly concave utility function case, recall the utility function properties in Section \ref{sec:Problem_formulation}, the utility function is positive $ U_i(r_{1i} + ... + r_{Ki}) > 0$, increasing and twice differentiable with respect to $r_{li}$. Then, it follows that $\frac{ \partial U_i(r_{1i} + ... + r_{Ki})}{\partial r_{li}} > 0$ and $\frac{\partial^2 U_i(r_{1i}+ ... + r_{Ki})}{\partial r_{li}^2} < 0$. It follows that, the utility function $\log(U_i(r_{1i} + r_{2i} + ... + r_{Ki}))$ in the optimization problem (\ref{eqn:opt_prob_fairness_mod}) have $\frac{\partial \log(U_i(r_{1i} + ... + r_{Ki}))}{\partial r_{li}} =  \frac{\frac{\partial U_i}{\partial r_{li}}}{U_i} > 0$ and $\frac{\partial ^2\log(U_i(r_{1i} + ... + r_{Ki}))}{\partial r_{li}^2} =  \frac{\frac{\partial^2 U_i}{\partial r_{li}^2}U_i-(\frac{\partial U_i}{\partial r_{li}})^2}{U^2_i} < 0$. Therefore, the strictly concave utility function $U_i(r_{1i} + r_{2i} + ... + r_{Ki})$ natural logarithm  $\log(U_i(r_{1i} + r_{2i} + ... + r_{Ki})
)$ is also strictly concave. It follows that the natural logarithm  of the logarithmic utility function in equation (\ref{eqn:log}) is strictly concave.

In the sigmoidal-like utility function case, the utility function of the normalized sigmoidal-like function is given by equation (\ref{eqn:sigmoid}) as $U_i(r_{1i} + r_{2i} + ... + r_{Ki}) = c_i\Big(\frac{1}{1+e^{-a_i(\sum_{l=1}^{K}r_{li}-b_i)}}-d_i\Big)$. For $0<\sum_{l=1}^{K}r_{li}<\sum_{l=1}^{K}R_l$, we have the first and second derivative as $\frac{\partial}{ \partial r_{li}}\log U_i(r_{1i} + ... + r_{Ki})  >0$ and $\frac{\partial^2}{\partial r_{li}^2}\log U_i(r_{1i} + ... + r_{Ki})  < 0$. Therefore, the sigmoidal-like utility function $U_i(r_{1i}+...+r_{Ki})$ natural logarithm  $\log(U_i(r_{1i}+...+r_{Ki}))$ is strictly concave function. Therefore, all the utility functions in our model have strictly concave natural logarithm.
\end{proof}
\begin{thm}\label{thm:global_soln}
The optimization problem (\ref{eqn:opt_prob_fairness}) is a convex optimization problem and there exists a unique tractable global optimal solution. 
\end{thm}
\begin{proof}
It follows from Lemma \ref{lem:concavity} that for all UEs utility functions are strictly concave. Therefore, the optimization problem (\ref{eqn:opt_prob_fairness_mod}) is a convex optimization problem \cite{Boyd2004}. The optimization problem (\ref{eqn:opt_prob_fairness_mod}) is equivalent to optimization problem (\ref{eqn:opt_prob_fairness}), therefore it is a convex optimization problem. For a convex optimization problem, there exists a unique tractable global optimal solution \cite{Boyd2004}.
\end{proof}

\section{The Dual Problem}\label{sec:Dual}

The key to a distributed and decentralized optimal solution of the primal problem in (\ref{eqn:opt_prob_fairness_mod}) is to convert it to the dual problem. We define the Lagrangian
\begin{equation}\label{eqn:lagrangian}
\begin{aligned}
& L(\textbf{r},\textbf{p}) =  \sum_{i=1}^{M}\log \Big(U_i(r_{1i} + r_{2i} + ... + r_{Ki})\Big)\\
& -p_1(\sum_{i=1}^{M}r_{1i} + z_{1i} - R_1) - ... - p_K(\sum_{i=1}^{M}r_{Ki} + z_{Ki} - R_K)
\end{aligned}
\end{equation}
where $z_{li}\geq 0$ is the $l^{th}$ slack variable of the $i^{th}$ constraint corresponding to $r_{li}$ and $p_l$ is the shadow price of the $l^{th}$ carrier eNodeB (i.e. the total price per unit rate for all the users in the coverage area of the $l^{th}$ carrier eNodeB) and $\textbf{p}=\{p_1,p_2,...,p_K\}$. Therefore, the $i^{th}$ UE bids for rate from the the $l^{th}$ carrier eNodeB can be written as $w_{li} = p_l r_{li}$ and we have $\sum_{i=1}^{M}w_{li} = p_l \sum_{i=1}^{M}r_{li}$. The dual problem objective function can be written as $D(\textbf{p}) =  \underset{{\textbf{r}}}\max \:L(\textbf{r},\textbf{p})$. Now, we divide the primal problem (\ref{eqn:opt_prob_fairness_mod}) into two simpler optimization problems in the UEs and the eNodeBs. The $i^{th}$ UE optimization problem is given by: 
\begin{equation}\label{eqn:opt_prob_fairness_UE}
\begin{aligned}
& \underset{{r_i}}{\text{max}}
& & \log(U_i(r_{1i} + r_{2i} + ... + r_{Ki}))-\sum_{l=1}^{K}p_lr_{li} \\
& \text{subject to}
& & p_l \geq 0\\
& & &  r_{li} \geq 0, \;\;\;\;\; i = 1,2, ...,M, l = 1,2, ...,K.
\end{aligned}
\end{equation}

The second problem is the $l^{th}$ eNodeB optimization problem for rate proportional fairness that is given by: 
\begin{equation}\label{eqn:opt_prob_fairness_eNodeB}
\begin{aligned}
& \underset{p_l}{\text{min}}
& & D(\textbf{p}) \\
& \text{subject to}
& & p_l \geq 0.\\
\end{aligned}
\end{equation}
The minimization of shadow price $p_l$ is achieved by the minimization of the slack variable $z_{li}$. Therefore, the maximum utilization of the $l^{th}$ eNodeB rate $R_l$ is achieved by setting the slack variable $z_{li} = 0$. The utility proportional fairness policy in the objective function of the optimization problem (\ref{eqn:opt_prob_fairness}) is guaranteed in the solution of the optimization problems (\ref{eqn:opt_prob_fairness_UE}) and (\ref{eqn:opt_prob_fairness_eNodeB}).

\section{Distributed Optimization Algorithm}\label{sec:Algorithm}

The distributed resource allocation algorithm for optimization problems (\ref{eqn:opt_prob_fairness_UE}) and (\ref{eqn:opt_prob_fairness_eNodeB}) is an iterative solution for allocating the network resources from multiple carriers simultaneously with utility proportional fairness policy. The algorithm is divided into the $i^{th}$ UE algorithm shown in Algorithm (\ref{alg:UE_FK}) and the $l^{th}$ eNodeB carrier algorithm shown in Algorithm (\ref{alg:eNodeB_FK}). In Algorithm (\ref{alg:UE_FK}) and (\ref{alg:eNodeB_FK}), all UEs bid for resources from the nearby eNodeBs (or eNodeB receives bids from all the UEs in its coverage). eNodeB sets a price for resource based on the sent UE bids. For UE with more than one nearby eNodeB, it chooses from the nearby carriers eNodeBs the one with the lowest shadow price and start requesting resources from that carrier eNodeB. If the allocated rate is not enough or the price of the resources increase due to high demand on that carrier eNodeB resources from other UEs, the UE 
switches to allocate the rest of the required resources from another nearby eNodeB carrier with a lower resource price. This is done iteratively until an equilibrium between demand and supply of resources is achieved and the optimal rates and price are allocated in the mobile network. Our distributed algorithm is set to avoid the situation of allocating zero rate to any user (i.e. no user is dropped) which is inherited from the utility proportional fairness policy in the optimization problem \ref{eqn:opt_prob_fairness}. 


\begin{algorithm}
\caption{The $i^{th}$ UE Algorithm}\label{alg:UE_FK}
\begin{algorithmic}
\STATE {Send initial bid $w_{li}(1)$ to $l^{th}$ carrier eNodeB (where $l \in L = \{1, 2, ..., K\}$)}
\LOOP
	\STATE {Receive shadow prices $p_{l\in L}(n)$ from all in range carriers eNodeBs}
	\IF {STOP from all in range carriers eNodeBs} %
		\STATE {Calculate allocated rates $r_{li} ^{\text{opt}}=\frac{w_{li}(n)}{p_l(n)}$}
		\STATE {STOP}
	\ELSE 
		\STATE{Set $p_{\min}^{0} = \{\}$ and $r_{i}^{0}=0$}
		\FOR{$m = 1 \to K$} 
		    \STATE{$p_{\min}^{m}(n) = \min (\textbf{p} \setminus \{p_{\min}^{0},p_{\min}^{1},...,p_{\min}^{m-1}\})$}
		    \STATE{$l_m = \{l \in L : p_{l}=\min (\textbf{p} \setminus \{p_{\min}^{0}, p_{\min}^{1}, ..., p_{\min}^{m-1}\}) \}$}
		    \COMMENT{$l_m$ is the index of the corresponding carrier}
		    \STATE {Solve $r_{l_mi}(n) = \arg \underset{r_{l_mi}}\max \Big(\log U_i(r_{1i}+ ... + r_{Ki}) - \sum_{l=1}^{K}p_l(n)r_{li}\Big)$ for the $l_m$ carrier eNodeB}
		    \STATE {$r_{i}^{m}(n) = r_{l_mi}(n)-\sum_{j=0}^{m-1}r_{i}^{j}(n)$}
		    \IF{$r_{i}^{m}(n)<0$} 
			      \STATE{Set $r_{i}^{m}(n)=0$} 
		    \ENDIF
		    \STATE {Send new bid $w_{l_mi} (n)= p_{\min}^{m}(n) r_{i}^{m}(n)$ to $l_m$ carrier eNodeB}
		\ENDFOR
	\ENDIF 
\ENDLOOP
\end{algorithmic}
\end{algorithm}


\begin{algorithm}
\caption{The $l^{th}$ eNodeB Algorithm}\label{alg:eNodeB_FK}
\begin{algorithmic}
\LOOP
	\STATE {Receive bids $w_{li}(n)$ from UEs}
	\COMMENT{Let $w_{li}(0) = 0\:\:\forall i$}
			\IF {$|w_{li}(n) -w_{li}(n-1)|<\delta  \:\:\forall i$} %
	   		\STATE {Allocate rates, $r_{li}^{\text{opt}}=\frac{w_{li}(n)}{p_l(n)}$ to $i^{th}$ UE}  
	   		\STATE {STOP} 
		\ELSE
	\STATE {Calculate $p_l(n) = \frac{\sum_{i=1}^{M}w_{li}(n)}{R_l}$}
	\STATE {Send new shadow price $p_l(n)$ to all UEs}
	\ENDIF 
\ENDLOOP
\end{algorithmic}
\end{algorithm}

\begin{figure}[tb]
\centering
\scalebox{0.5}{ 
\begin{pspicture}(0,-4.44)(18.0,4.44)
\definecolor{color3787}{rgb}{0.9607843137254902,0.12941176470588237,0.07058823529411765}
\definecolor{color3788}{rgb}{0.9725490196078431,0.058823529411764705,0.058823529411764705}
\definecolor{color3789}{rgb}{0.9333333333333333,0.06666666666666667,0.06666666666666667}
\definecolor{color3835}{rgb}{0.0392156862745098,0.054901960784313725,0.9529411764705882}
\definecolor{color3836}{rgb}{0.0784313725490196,0.1803921568627451,0.9450980392156862}
\definecolor{color3837}{rgb}{0.0784313725490196,0.0784313725490196,0.9254901960784314}
\definecolor{color3865}{rgb}{0.09803921568627451,0.47058823529411764,0.09411764705882353}
\definecolor{color3866}{rgb}{0.08235294117647059,0.48627450980392156,0.08627450980392157}
\definecolor{color3867}{rgb}{0.08235294117647059,0.5725490196078431,0.09411764705882353}
\psframe[linewidth=0.04,linecolor=color3787,dimen=outer,fillstyle=solid](4.42,3.3076692)(3.86,2.34)
\psframe[linewidth=0.04,linecolor=color3788,dimen=outer,fillstyle=solid](4.3458824,3.2001505)(3.942353,2.8873684)
\psframe[linewidth=0.04,linecolor=color3789,dimen=outer,fillstyle=solid](3.9588234,3.64)(3.86,3.2783458)
\psframe[linewidth=0.04,linecolor=color3787,dimen=outer,fillstyle=solid](3.96,0.48766917)(3.4,-0.48)
\psframe[linewidth=0.04,linecolor=color3788,dimen=outer,fillstyle=solid](3.8858824,0.38015038)(3.482353,0.06736842)
\psframe[linewidth=0.04,linecolor=color3789,dimen=outer,fillstyle=solid](3.4988236,0.82)(3.4,0.45834586)
\psframe[linewidth=0.04,linecolor=color3787,dimen=outer,fillstyle=solid](6.22,-1.7323308)(5.66,-2.7)
\psframe[linewidth=0.04,linecolor=color3788,dimen=outer,fillstyle=solid](6.145882,-1.8398496)(5.742353,-2.1526315)
\psframe[linewidth=0.04,linecolor=color3789,dimen=outer,fillstyle=solid](5.7588234,-1.4)(5.66,-1.7616541)
\psframe[linewidth=0.04,linecolor=color3787,dimen=outer,fillstyle=solid](1.32,1.3876692)(0.76,0.42)
\psframe[linewidth=0.04,linecolor=color3788,dimen=outer,fillstyle=solid](1.2458824,1.2801504)(0.8423529,0.9673684)
\psframe[linewidth=0.04,linecolor=color3789,dimen=outer,fillstyle=solid](0.85882354,1.72)(0.76,1.3583459)
\psframe[linewidth=0.04,linecolor=color3787,dimen=outer,fillstyle=solid](2.18,-0.81233084)(1.62,-1.78)
\psframe[linewidth=0.04,linecolor=color3788,dimen=outer,fillstyle=solid](2.1058824,-0.91984963)(1.702353,-1.2326316)
\psframe[linewidth=0.04,linecolor=color3789,dimen=outer,fillstyle=solid](1.7188236,-0.48)(1.62,-0.8416541)
\psframe[linewidth=0.04,linecolor=color3787,dimen=outer,fillstyle=solid](4.1,-2.592331)(3.54,-3.56)
\psframe[linewidth=0.04,linecolor=color3788,dimen=outer,fillstyle=solid](4.0258822,-2.6998496)(3.6223528,-3.0126317)
\psframe[linewidth=0.04,linecolor=color3789,dimen=outer,fillstyle=solid](3.6388235,-2.26)(3.54,-2.621654)
\psellipse[linewidth=0.04,dimen=outer](5.91,-0.01)(5.91,4.43)
\psellipse[linewidth=0.04,dimen=outer](12.05,0.04)(5.95,4.4)
\psframe[linewidth=0.04,dimen=outer](13.0,1.12)(11.98,-0.38)
\psline[linewidth=0.04cm](12.498361,2.12)(12.515082,1.103871)
\psline[linewidth=0.04cm](12.130492,1.636129)(12.816066,1.636129)
\psline[linewidth=0.04cm](12.832787,2.12)(12.832787,1.62)
\psline[linewidth=0.04cm](12.147213,2.12)(12.147213,1.636129)
\psframe[linewidth=0.04,dimen=outer](5.96,1.18)(4.94,-0.32)
\psline[linewidth=0.04cm](5.4583607,2.18)(5.475082,1.1638709)
\psline[linewidth=0.04cm](5.090492,1.6961291)(5.7760653,1.6961291)
\psline[linewidth=0.04cm](5.792787,2.18)(5.792787,1.68)
\psline[linewidth=0.04cm](5.107213,2.18)(5.107213,1.6961291)
\psframe[linewidth=0.04,linecolor=color3835,dimen=outer,fillstyle=solid](15.32,0.7476697)(14.76,-0.22)
\psframe[linewidth=0.04,linecolor=color3836,dimen=outer,fillstyle=solid](15.245883,0.6401501)(14.842352,0.32736817)
\psframe[linewidth=0.04,linecolor=color3837,dimen=outer,fillstyle=solid](14.858823,1.08)(14.76,0.71834594)
\psframe[linewidth=0.04,linecolor=color3835,dimen=outer,fillstyle=solid](17.02,-0.25233033)(16.46,-1.22)
\psframe[linewidth=0.04,linecolor=color3836,dimen=outer,fillstyle=solid](16.945883,-0.35984984)(16.542353,-0.67263186)
\psframe[linewidth=0.04,linecolor=color3837,dimen=outer,fillstyle=solid](16.558823,0.08)(16.46,-0.28165406)
\psframe[linewidth=0.04,linecolor=color3835,dimen=outer,fillstyle=solid](14.38,-0.6923303)(13.82,-1.66)
\psframe[linewidth=0.04,linecolor=color3836,dimen=outer,fillstyle=solid](14.305882,-0.79984987)(13.902352,-1.1126318)
\psframe[linewidth=0.04,linecolor=color3837,dimen=outer,fillstyle=solid](13.918823,-0.36)(13.82,-0.72165406)
\psframe[linewidth=0.04,linecolor=color3835,dimen=outer,fillstyle=solid](12.12,-2.7123303)(11.56,-3.68)
\psframe[linewidth=0.04,linecolor=color3836,dimen=outer,fillstyle=solid](12.045882,-2.81985)(11.642352,-3.1326318)
\psframe[linewidth=0.04,linecolor=color3837,dimen=outer,fillstyle=solid](11.658823,-2.38)(11.56,-2.7416542)
\psframe[linewidth=0.04,linecolor=color3835,dimen=outer,fillstyle=solid](12.04,3.8276696)(11.48,2.86)
\psframe[linewidth=0.04,linecolor=color3836,dimen=outer,fillstyle=solid](11.965882,3.7201502)(11.562352,3.4073682)
\psframe[linewidth=0.04,linecolor=color3837,dimen=outer,fillstyle=solid](11.578823,4.16)(11.48,3.798346)
\psframe[linewidth=0.04,linecolor=color3835,dimen=outer,fillstyle=solid](15.24,2.7876697)(14.68,1.82)
\psframe[linewidth=0.04,linecolor=color3836,dimen=outer,fillstyle=solid](15.165882,2.68015)(14.762352,2.3673682)
\psframe[linewidth=0.04,linecolor=color3837,dimen=outer,fillstyle=solid](14.778823,3.12)(14.68,2.7583458)
\psframe[linewidth=0.04,linecolor=color3865,dimen=outer,fillstyle=solid](9.0,0.7676697)(8.44,-0.2)
\psframe[linewidth=0.04,linecolor=color3866,dimen=outer,fillstyle=solid](8.925882,0.6601502)(8.522352,0.34736815)
\psframe[linewidth=0.04,linecolor=color3867,dimen=outer,fillstyle=solid](8.538823,1.1)(8.44,0.7383459)
\psframe[linewidth=0.04,linecolor=color3865,dimen=outer,fillstyle=solid](8.02,-0.53233033)(7.46,-1.5)
\psframe[linewidth=0.04,linecolor=color3866,dimen=outer,fillstyle=solid](7.945883,-0.63984984)(7.542352,-0.95263183)
\psframe[linewidth=0.04,linecolor=color3867,dimen=outer,fillstyle=solid](7.558823,-0.2)(7.46,-0.56165403)
\psframe[linewidth=0.04,linecolor=color3865,dimen=outer,fillstyle=solid](10.34,-1.3123304)(9.78,-2.28)
\psframe[linewidth=0.04,linecolor=color3866,dimen=outer,fillstyle=solid](10.2658825,-1.4198499)(9.862352,-1.7326318)
\psframe[linewidth=0.04,linecolor=color3867,dimen=outer,fillstyle=solid](9.878823,-0.98)(9.78,-1.3416541)
\psframe[linewidth=0.04,linecolor=color3865,dimen=outer,fillstyle=solid](7.52,1.5676696)(6.96,0.6)
\psframe[linewidth=0.04,linecolor=color3866,dimen=outer,fillstyle=solid](7.445883,1.4601501)(7.042352,1.1473682)
\psframe[linewidth=0.04,linecolor=color3867,dimen=outer,fillstyle=solid](7.058823,1.9)(6.96,1.5383459)
\psframe[linewidth=0.04,linecolor=color3865,dimen=outer,fillstyle=solid](9.08,3.0276697)(8.52,2.06)
\psframe[linewidth=0.04,linecolor=color3866,dimen=outer,fillstyle=solid](9.005882,2.92015)(8.602352,2.6073682)
\psframe[linewidth=0.04,linecolor=color3867,dimen=outer,fillstyle=solid](8.618823,3.36)(8.52,2.9983459)
\psframe[linewidth=0.04,linecolor=color3865,dimen=outer,fillstyle=solid](10.98,0.60766965)(10.42,-0.36)
\psframe[linewidth=0.04,linecolor=color3866,dimen=outer,fillstyle=solid](10.905883,0.50015014)(10.502353,0.18736817)
\psframe[linewidth=0.04,linecolor=color3867,dimen=outer,fillstyle=solid](10.518824,0.94)(10.42,0.57834595)
\usefont{T1}{ptm}{m}{n}
\rput(4.137344,2.07){\large 1}
\usefont{T1}{ptm}{m}{n}
\rput(1.03625,0.11){\large 2}
\usefont{T1}{ptm}{m}{n}
\rput(1.8789062,-2.09){\large 4}
\usefont{T1}{ptm}{m}{n}
\rput(3.6828125,-0.75){\large 3}
\usefont{T1}{ptm}{m}{n}
\rput(3.845,-3.83){\large 5}
\usefont{T1}{ptm}{b}{n}
\rput(5.484375,-0.65){\large C1}
\usefont{T1}{ptm}{m}{n}
\rput(15.031875,-0.51){\large 9}
\usefont{T1}{ptm}{m}{n}
\rput(11.852032,2.45){\large 7}
\usefont{T1}{ptm}{m}{n}
\rput(14.964063,1.47){\large 8}
\usefont{T1}{ptm}{m}{n}
\rput(16.686407,-1.51){\large 10}
\usefont{T1}{ptm}{m}{n}
\rput(11.825156,-4.01){\large 12}
\usefont{T1}{ptm}{m}{n}
\rput(5.972031,-3.03){\large 6}
\usefont{T1}{ptm}{m}{n}
\rput(14.067344,-1.95){\large 11}
\usefont{T1}{ptm}{m}{n}
\rput(8.775937,1.71){\large 13}
\usefont{T1}{ptm}{m}{n}
\rput(7.2265625,0.29){\large 14}
\usefont{T1}{ptm}{m}{n}
\rput(7.744375,-1.81){\large 17}
\usefont{T1}{ptm}{m}{n}
\rput(8.6857815,-0.49){\large 16}
\usefont{T1}{ptm}{m}{n}
\rput(10.719687,-0.69){\large 15}
\usefont{T1}{ptm}{m}{n}
\rput(9.981563,-2.57){\large 18}
\usefont{T1}{ptm}{b}{n}
\rput(12.4704685,-0.69){\large C2}
\end{pspicture} 
} 
\caption{System Model with three groups of users. The $1^{st}$ group with UE indexes $i = 1,2,3,4,5,6$ (red), $2^{nd}$ group with UE indexes $i = 7,8,9,10,11,12$ (blue), and $3^{rd}$ group with UE indexes $i = 13,14,15,16,17,18$ (green).}
\label{fig:sim:System_Model}
\end{figure}
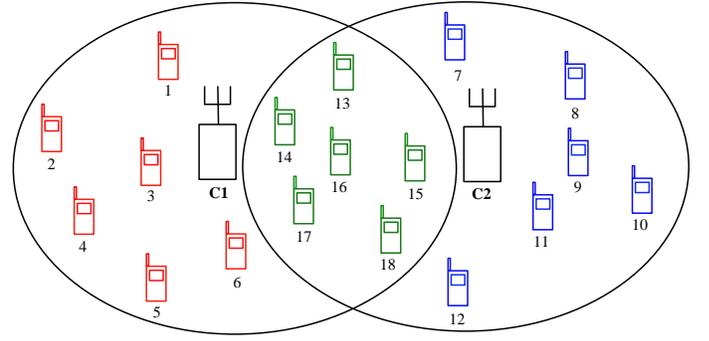

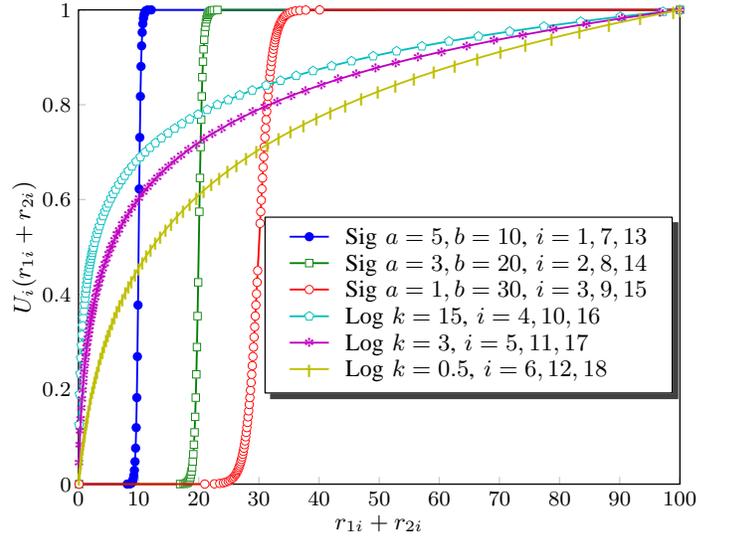
\begin{figure}[tb]
\centering
%
\psset{xunit=0.010000\plotwidth,yunit=0.788710\plotwidth}%
\begin{pspicture}(-11.059908,-0.111111)(102.764977,1.023392)%
\psline[linewidth=\AxesLineWidth,linecolor=GridColor](0.000000,0.000000)(0.000000,0.015215)
\psline[linewidth=\AxesLineWidth,linecolor=GridColor](10.000000,0.000000)(10.000000,0.015215)
\psline[linewidth=\AxesLineWidth,linecolor=GridColor](20.000000,0.000000)(20.000000,0.015215)
\psline[linewidth=\AxesLineWidth,linecolor=GridColor](30.000000,0.000000)(30.000000,0.015215)
\psline[linewidth=\AxesLineWidth,linecolor=GridColor](40.000000,0.000000)(40.000000,0.015215)
\psline[linewidth=\AxesLineWidth,linecolor=GridColor](50.000000,0.000000)(50.000000,0.015215)
\psline[linewidth=\AxesLineWidth,linecolor=GridColor](60.000000,0.000000)(60.000000,0.015215)
\psline[linewidth=\AxesLineWidth,linecolor=GridColor](70.000000,0.000000)(70.000000,0.015215)
\psline[linewidth=\AxesLineWidth,linecolor=GridColor](80.000000,0.000000)(80.000000,0.015215)
\psline[linewidth=\AxesLineWidth,linecolor=GridColor](90.000000,0.000000)(90.000000,0.015215)
\psline[linewidth=\AxesLineWidth,linecolor=GridColor](100.000000,0.000000)(100.000000,0.015215)
\psline[linewidth=\AxesLineWidth,linecolor=GridColor](0.000000,0.000000)(1.200000,0.000000)
\psline[linewidth=\AxesLineWidth,linecolor=GridColor](0.000000,0.200000)(1.200000,0.200000)
\psline[linewidth=\AxesLineWidth,linecolor=GridColor](0.000000,0.400000)(1.200000,0.400000)
\psline[linewidth=\AxesLineWidth,linecolor=GridColor](0.000000,0.600000)(1.200000,0.600000)
\psline[linewidth=\AxesLineWidth,linecolor=GridColor](0.000000,0.800000)(1.200000,0.800000)
\psline[linewidth=\AxesLineWidth,linecolor=GridColor](0.000000,1.000000)(1.200000,1.000000)
{ \footnotesize 
\rput[t](0.000000,-0.015215){$0$}
\rput[t](10.000000,-0.015215){$10$}
\rput[t](20.000000,-0.015215){$20$}
\rput[t](30.000000,-0.015215){$30$}
\rput[t](40.000000,-0.015215){$40$}
\rput[t](50.000000,-0.015215){$50$}
\rput[t](60.000000,-0.015215){$60$}
\rput[t](70.000000,-0.015215){$70$}
\rput[t](80.000000,-0.015215){$80$}
\rput[t](90.000000,-0.015215){$90$}
\rput[t](100.000000,-0.015215){$100$}
\rput[r](-1.200000,0.000000){$0$}
\rput[r](-1.200000,0.200000){$0.2$}
\rput[r](-1.200000,0.400000){$0.4$}
\rput[r](-1.200000,0.600000){$0.6$}
\rput[r](-1.200000,0.800000){$0.8$}
\rput[r](-1.200000,1.000000){$1$}
} 
\psframe[linewidth=\AxesLineWidth,dimen=middle](0.000000,0.000000)(100.000000,1.000000)
{ \small 
\rput[b](50.000000,-0.111111){
\begin{tabular}{c}
$r_{1i} + r_{2i}$\\
\end{tabular}
}
\rput[t]{90}(-11.059908,0.500000){
\begin{tabular}{c}
$U_i(r_{1i} + r_{2i})$\\
\end{tabular}
}
} 
\newrgbcolor{color18.0052}{0  0  1}
\psline[plotstyle=line,linejoin=1,showpoints=false,dotstyle=*,dotsize=\MarkerSize,linestyle=solid,linewidth=\LineWidth,linecolor=color18.0052]
(100.000000,1.000000)(100.000000,1.000000)
\psline[plotstyle=line,linejoin=1,showpoints=true,dotstyle=*,dotsize=\MarkerSize,linestyle=solid,linewidth=\LineWidth,linecolor=color18.0052](0.100000,0.000000)(8.100000,0.000075)(8.400000,0.000335)(8.600000,0.000911)(8.700000,0.001501)
(8.800000,0.002473)(8.900000,0.004070)(9.000000,0.006693)(9.100000,0.010987)(9.200000,0.017986)
(9.300000,0.029312)(9.400000,0.047426)(9.500000,0.075858)(9.600000,0.119203)(9.700000,0.182426)
(9.800000,0.268941)(9.900000,0.377541)(10.100000,0.622459)(10.200000,0.731059)(10.300000,0.817574)
(10.400000,0.880797)(10.500000,0.924142)(10.600000,0.952574)(10.700000,0.970688)(10.800000,0.982014)
(10.900000,0.989013)(11.000000,0.993307)(11.100000,0.995930)(11.200000,0.997527)(11.300000,0.998499)
(11.400000,0.999089)(11.600000,0.999665)(12.100000,0.999972)(100.000000,1.000000)
\newrgbcolor{color19.0048}{0         0.5           0}
\psline[plotstyle=line,linejoin=1,showpoints=false,dotstyle=Bsquare,dotsize=\MarkerSize,linestyle=solid,linewidth=\LineWidth,linecolor=color19.0048]
(100.000000,1.000000)(100.000000,1.000000)
\psline[plotstyle=line,linejoin=1,showpoints=true,dotstyle=Bsquare,dotsize=\MarkerSize,linestyle=solid,linewidth=\LineWidth,linecolor=color19.0048]
(0.100000,0.000000)(16.900000,0.000091)(17.500000,0.000553)(17.800000,0.001359)(18.000000,0.002473)
(18.100000,0.003335)(18.200000,0.004496)(18.300000,0.006060)(18.400000,0.008163)(18.500000,0.010987)
(18.600000,0.014774)(18.700000,0.019840)(18.800000,0.026597)(18.900000,0.035571)(19.000000,0.047426)
(19.100000,0.062973)(19.200000,0.083173)(19.300000,0.109097)(19.400000,0.141851)(19.500000,0.182426)
(19.600000,0.231475)(19.700000,0.289050)(19.800000,0.354344)(19.900000,0.425557)(20.100000,0.574443)
(20.200000,0.645656)(20.300000,0.710950)(20.400000,0.768525)(20.500000,0.817574)(20.600000,0.858149)
(20.700000,0.890903)(20.800000,0.916827)(20.900000,0.937027)(21.000000,0.952574)(21.100000,0.964429)
(21.200000,0.973403)(21.300000,0.980160)(21.400000,0.985226)(21.500000,0.989013)(21.600000,0.991837)
(21.700000,0.993940)(21.800000,0.995504)(21.900000,0.996665)(22.000000,0.997527)(22.200000,0.998641)
(22.500000,0.999447)(23.100000,0.999909)(100.000000,1.000000)
\newrgbcolor{color20.0048}{1  0  0}
\psline[plotstyle=line,linejoin=1,showpoints=false,dotstyle=Bo,dotsize=\MarkerSize,linestyle=solid,linewidth=\LineWidth,linecolor=color20.0048]
(100.000000,1.000000)(100.000000,1.000000)
\psline[plotstyle=line,linejoin=1,showpoints=true,dotstyle=Bo,dotsize=\MarkerSize,linestyle=solid,linewidth=\LineWidth,linecolor=color20.0048]
(0.100000,0.000000)(21.000000,0.000123)(22.600000,0.000611)(23.500000,0.001501)(24.100000,0.002732)
(24.500000,0.004070)(24.900000,0.006060)(25.200000,0.008163)(25.500000,0.010987)(25.700000,0.013387)
(25.900000,0.016302)(26.100000,0.019840)(26.200000,0.021881)(26.300000,0.024127)(26.400000,0.026597)
(26.500000,0.029312)(26.600000,0.032295)(26.700000,0.035571)(26.800000,0.039166)(26.900000,0.043107)
(27.000000,0.047426)(27.100000,0.052154)(27.200000,0.057324)(27.300000,0.062973)(27.400000,0.069138)
(27.500000,0.075858)(27.600000,0.083173)(27.700000,0.091123)(27.800000,0.099750)(27.900000,0.109097)
(28.000000,0.119203)(28.100000,0.130108)(28.200000,0.141851)(28.300000,0.154465)(28.400000,0.167982)
(28.500000,0.182426)(28.600000,0.197816)(28.700000,0.214165)(28.800000,0.231475)(28.900000,0.249740)
(29.000000,0.268941)(29.100000,0.289050)(29.200000,0.310026)(29.300000,0.331812)(29.400000,0.354344)
(29.500000,0.377541)(29.600000,0.401312)(29.700000,0.425557)(29.800000,0.450166)(30.200000,0.549834)
(30.300000,0.574443)(30.400000,0.598688)(30.500000,0.622459)(30.600000,0.645656)(30.700000,0.668188)
(30.800000,0.689974)(30.900000,0.710950)(31.000000,0.731059)(31.100000,0.750260)(31.200000,0.768525)
(31.300000,0.785835)(31.400000,0.802184)(31.500000,0.817574)(31.600000,0.832018)(31.700000,0.845535)
(31.800000,0.858149)(31.900000,0.869892)(32.000000,0.880797)(32.100000,0.890903)(32.200000,0.900250)
(32.300000,0.908877)(32.400000,0.916827)(32.500000,0.924142)(32.600000,0.930862)(32.700000,0.937027)
(32.800000,0.942676)(32.900000,0.947846)(33.000000,0.952574)(33.100000,0.956893)(33.200000,0.960834)
(33.300000,0.964429)(33.400000,0.967705)(33.500000,0.970688)(33.600000,0.973403)(33.700000,0.975873)
(33.800000,0.978119)(34.000000,0.982014)(34.200000,0.985226)(34.400000,0.987872)(34.600000,0.990048)
(34.900000,0.992608)(35.200000,0.994514)(35.600000,0.996316)(36.100000,0.997762)(36.800000,0.998887)
(37.800000,0.999590)(40.100000,0.999959)(100.000000,1.000000)
\newrgbcolor{color21.0048}{0        0.75        0.75}
\psline[plotstyle=line,linejoin=1,showpoints=false,dotstyle=Bpentagon,dotsize=\MarkerSize,linestyle=solid,linewidth=\LineWidth,linecolor=color21.0048]
(100.000000,1.000000)(100.000000,1.000000)
\psline[plotstyle=line,linejoin=1,showpoints=true,dotstyle=Bpentagon,dotsize=\MarkerSize,linestyle=solid,linewidth=\LineWidth,linecolor=color21.0048]
(0.100000,0.125281)(0.200000,0.189543)(0.300000,0.233084)(0.400000,0.266057)(0.500000,0.292603)
(0.600000,0.314824)(0.700000,0.333933)(0.800000,0.350696)(0.900000,0.365626)(1.000000,0.379086)
(1.100000,0.391338)(1.200000,0.402582)(1.300000,0.412971)(1.400000,0.422627)(1.500000,0.431645)
(1.600000,0.440105)(1.700000,0.448072)(1.800000,0.455600)(1.900000,0.462735)(2.000000,0.469516)
(2.100000,0.475977)(2.200000,0.482146)(2.300000,0.488049)(2.400000,0.493707)(2.500000,0.499141)
(2.700000,0.509400)(2.900000,0.518943)(3.100000,0.527863)(3.300000,0.536237)(3.500000,0.544127)
(3.800000,0.555169)(4.100000,0.565386)(4.400000,0.574892)(4.700000,0.583780)(5.000000,0.592125)
(5.400000,0.602514)(5.800000,0.612169)(6.200000,0.621187)(6.700000,0.631683)(7.200000,0.641430)
(7.700000,0.650528)(8.300000,0.660703)(8.900000,0.670172)(9.600000,0.680450)(10.300000,0.690009)
(11.100000,0.700173)(11.900000,0.709633)(12.800000,0.719548)(13.800000,0.729781)(14.900000,0.740219)
(16.000000,0.749915)(17.200000,0.759764)(18.500000,0.769689)(19.900000,0.779628)(21.400000,0.789533)
(23.100000,0.799953)(24.900000,0.810184)(26.800000,0.820212)(28.900000,0.830502)(31.200000,0.840949)
(33.600000,0.851060)(36.200000,0.861231)(39.000000,0.871400)(42.100000,0.881840)(45.400000,0.892143)
(49.000000,0.902561)(52.900000,0.913018)(57.100000,0.923452)(61.600000,0.933812)(66.500000,0.944266)
(71.700000,0.954550)(77.400000,0.965000)(83.500000,0.975363)(90.100000,0.985756)(97.200000,0.996120)
(100.000000,1.000000)
\newrgbcolor{color22.0046}{0.75           0        0.75}
\psline[plotstyle=line,linejoin=1,showpoints=false,dotstyle=Basterisk,dotsize=\MarkerSize,linestyle=solid,linewidth=\LineWidth,linecolor=color22.0046]
(100.000000,1.000000)(100.000000,1.000000)
\psline[plotstyle=line,linejoin=1,showpoints=true,dotstyle=Basterisk,dotsize=\MarkerSize,linestyle=solid,linewidth=\LineWidth,linecolor=color22.0046]
(0.100000,0.045971)(0.200000,0.082354)(0.300000,0.112466)(0.400000,0.138154)(0.500000,0.160552)
(0.600000,0.180410)(0.700000,0.198244)(0.800000,0.214430)(0.900000,0.229246)(1.000000,0.242907)
(1.100000,0.255579)(1.200000,0.267396)(1.300000,0.278466)(1.400000,0.288878)(1.500000,0.298706)
(1.600000,0.308012)(1.700000,0.316848)(1.800000,0.325261)(1.900000,0.333287)(2.000000,0.340962)
(2.100000,0.348315)(2.200000,0.355372)(2.300000,0.362156)(2.400000,0.368686)(2.500000,0.374982)
(2.600000,0.381060)(2.800000,0.392617)(3.000000,0.403459)(3.200000,0.413669)(3.400000,0.423316)
(3.600000,0.432460)(3.800000,0.441151)(4.100000,0.453428)(4.400000,0.464901)(4.700000,0.475669)
(5.000000,0.485813)(5.300000,0.495402)(5.600000,0.504493)(6.000000,0.515925)(6.400000,0.526656)
(6.800000,0.536767)(7.300000,0.548638)(7.800000,0.559755)(8.300000,0.570209)(8.900000,0.581981)
(9.500000,0.593013)(10.100000,0.603391)(10.800000,0.614769)(11.500000,0.625454)(12.300000,0.636916)
(13.100000,0.647675)(14.000000,0.659038)(15.000000,0.670855)(16.000000,0.681925)(17.100000,0.693345)
(18.300000,0.705009)(19.600000,0.716826)(20.900000,0.727896)(22.300000,0.739084)(23.800000,0.750328)
(25.400000,0.761576)(27.100000,0.772785)(29.000000,0.784519)(31.000000,0.796076)(33.100000,0.807443)
(35.400000,0.819100)(37.800000,0.830491)(40.400000,0.842048)(43.200000,0.853696)(46.200000,0.865373)
(49.400000,0.877027)(52.800000,0.888614)(56.500000,0.900409)(60.400000,0.912039)(64.600000,0.923755)
(69.100000,0.935496)(73.900000,0.947209)(79.000000,0.958851)(84.500000,0.970596)(90.400000,0.982378)
(96.700000,0.994140)(100.000000,1.000000)
\newrgbcolor{color23.0046}{0.75        0.75           0}
\psline[plotstyle=line,linejoin=1,showpoints=false,dotstyle=B|,dotsize=\MarkerSize,linestyle=solid,linewidth=\LineWidth,linecolor=color23.0046]
(100.000000,1.000000)(100.000000,1.000000)
\psline[plotstyle=line,linejoin=1,showpoints=true,dotstyle=B|,dotsize=\MarkerSize,linestyle=solid,linewidth=\LineWidth,linecolor=color23.0046]
(0.100000,0.012409)(0.200000,0.024241)(0.300000,0.035546)(0.400000,0.046371)(0.500000,0.056753)
(0.600000,0.066728)(0.700000,0.076327)(0.800000,0.085577)(0.900000,0.094502)(1.000000,0.103124)
(1.100000,0.111463)(1.200000,0.119538)(1.300000,0.127365)(1.400000,0.134957)(1.500000,0.142330)
(1.700000,0.156463)(1.900000,0.169852)(2.100000,0.182572)(2.300000,0.194685)(2.500000,0.206248)
(2.700000,0.217308)(2.900000,0.227906)(3.100000,0.238081)(3.400000,0.252618)(3.700000,0.266370)
(4.000000,0.279415)(4.300000,0.291824)(4.600000,0.303656)(4.900000,0.314962)(5.200000,0.325786)
(5.600000,0.339537)(6.000000,0.352583)(6.400000,0.364992)(6.800000,0.376824)(7.300000,0.390879)
(7.800000,0.404198)(8.300000,0.416854)(8.800000,0.428910)(9.400000,0.442661)(10.000000,0.455707)
(10.600000,0.468116)(11.300000,0.481867)(12.000000,0.494913)(12.800000,0.509046)(13.600000,0.522435)
(14.500000,0.536701)(15.400000,0.550208)(16.400000,0.564421)(17.400000,0.577881)(18.500000,0.591908)
(19.600000,0.605201)(20.800000,0.618953)(22.100000,0.633056)(23.400000,0.646418)(24.800000,0.660064)
(26.300000,0.673915)(27.900000,0.687902)(29.600000,0.701967)(31.400000,0.716056)(33.300000,0.730128)
(35.300000,0.744144)(37.400000,0.758075)(39.600000,0.771894)(41.900000,0.785581)(44.400000,0.799667)
(47.000000,0.813534)(49.800000,0.827667)(52.700000,0.841522)(55.800000,0.855542)(59.100000,0.869663)
(62.600000,0.883830)(66.300000,0.897996)(70.200000,0.912119)(74.300000,0.926167)(78.700000,0.940426)
(83.300000,0.954525)(88.200000,0.968731)(93.400000,0.982986)(98.900000,0.997242)(100.000000,1.000000)
{ \small 
\rput(65.349462,0.371726){%
\psshadowbox[framesep=0pt,linewidth=\AxesLineWidth]{\psframebox*{\begin{tabular}{l}
\Rnode{a1}{\hspace*{0.0ex}} \hspace*{0.3cm} \Rnode{a2}{~~Sig $a=5, b=10$, $i=1,7,13$} \\
\Rnode{a3}{\hspace*{0.0ex}} \hspace*{0.3cm} \Rnode{a4}{~~Sig $a=3, b=20$, $i=2,8,14$} \\
\Rnode{a5}{\hspace*{0.0ex}} \hspace*{0.3cm} \Rnode{a6}{~~Sig $a=1, b=30$, $i=3,9,15$} \\
\Rnode{a7}{\hspace*{0.0ex}} \hspace*{0.3cm} \Rnode{a8}{~~Log $k=15$, $i=4,10,16$} \\
\Rnode{a9}{\hspace*{0.0ex}} \hspace*{0.3cm} \Rnode{a10}{~~Log $k=3$, $i=5,11,17$} \\
\Rnode{a11}{\hspace*{0.0ex}} \hspace*{0.3cm} \Rnode{a12}{~~Log $k=0.5$, $i=6,12,18$} \\
\end{tabular}}
\ncline[linestyle=solid,linewidth=\LineWidth,linecolor=color18.0052]{a1}{a2}
\ncput{\psdot[dotstyle=*,dotsize=\MarkerSize,linecolor=color18.0052]}
\ncline[linestyle=solid,linewidth=\LineWidth,linecolor=color19.0048]{a3}{a4}
\ncput{\psdot[dotstyle=Bsquare,dotsize=\MarkerSize,linecolor=color19.0048]}
\ncline[linestyle=solid,linewidth=\LineWidth,linecolor=color20.0048]{a5}{a6}
\ncput{\psdot[dotstyle=Bo,dotsize=\MarkerSize,linecolor=color20.0048]}
\ncline[linestyle=solid,linewidth=\LineWidth,linecolor=color21.0048]{a7}{a8}
\ncput{\psdot[dotstyle=Bpentagon,dotsize=\MarkerSize,linecolor=color21.0048]}
\ncline[linestyle=solid,linewidth=\LineWidth,linecolor=color22.0046]{a9}{a10}
\ncput{\psdot[dotstyle=Basterisk,dotsize=\MarkerSize,linecolor=color22.0046]}
\ncline[linestyle=solid,linewidth=\LineWidth,linecolor=color23.0046]{a11}{a12}
\ncput{\psdot[dotstyle=B|,dotsize=\MarkerSize,linecolor=color23.0046]}
}%
}%
} 
\end{pspicture}%

\caption{The users utility functions $U_i(r_{1i}+r_{2i})$ used in the simulation (three sigmoidal-like functions and three logarithmic functions).}
\label{fig:sim:Utilities}
\end{figure}

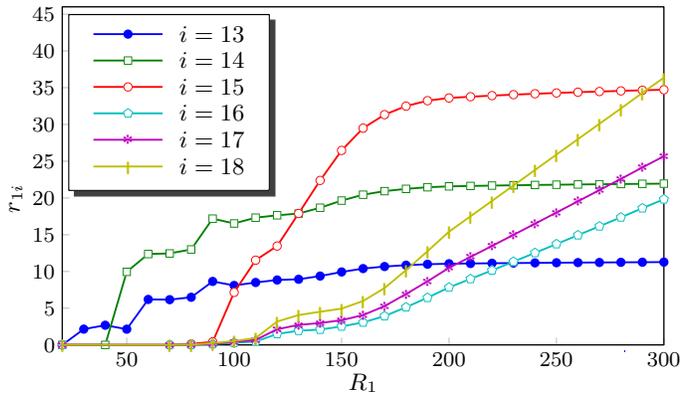
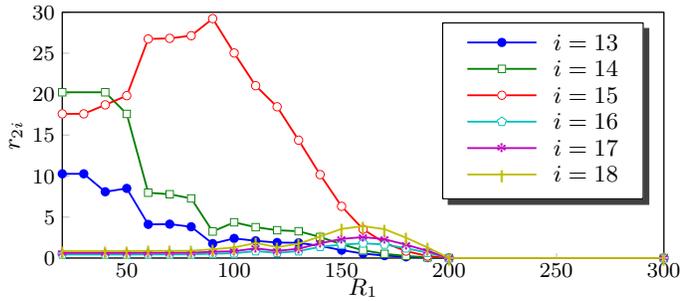
\begin{figure}[tb]
\centering
\subfigure[The allocated rates $r_{1i}$ from the $1^{st}$ carrier eNodeB to the $3^{rd}$ group of users.]{
\label{fig:sim:rates31}
%
\psset{xunit=0.003571\plotwidth,yunit=0.012218\plotwidth}%
\begin{pspicture}(-8.387097,-7.169811)(307.741935,46.509434)%


\psline[linewidth=\AxesLineWidth,linecolor=GridColor](50.000000,0.000000)(50.000000,0.982178)
\psline[linewidth=\AxesLineWidth,linecolor=GridColor](100.000000,0.000000)(100.000000,0.982178)
\psline[linewidth=\AxesLineWidth,linecolor=GridColor](150.000000,0.000000)(150.000000,0.982178)
\psline[linewidth=\AxesLineWidth,linecolor=GridColor](200.000000,0.000000)(200.000000,0.982178)
\psline[linewidth=\AxesLineWidth,linecolor=GridColor](250.000000,0.000000)(250.000000,0.982178)
\psline[linewidth=\AxesLineWidth,linecolor=GridColor](300.000000,0.000000)(300.000000,0.982178)
\psline[linewidth=\AxesLineWidth,linecolor=GridColor](20.000000,0.000000)(23.360000,0.000000)
\psline[linewidth=\AxesLineWidth,linecolor=GridColor](20.000000,5.000000)(23.360000,5.000000)
\psline[linewidth=\AxesLineWidth,linecolor=GridColor](20.000000,10.000000)(23.360000,10.000000)
\psline[linewidth=\AxesLineWidth,linecolor=GridColor](20.000000,15.000000)(23.360000,15.000000)
\psline[linewidth=\AxesLineWidth,linecolor=GridColor](20.000000,20.000000)(23.360000,20.000000)
\psline[linewidth=\AxesLineWidth,linecolor=GridColor](20.000000,25.000000)(23.360000,25.000000)
\psline[linewidth=\AxesLineWidth,linecolor=GridColor](20.000000,30.000000)(23.360000,30.000000)
\psline[linewidth=\AxesLineWidth,linecolor=GridColor](20.000000,35.000000)(23.360000,35.000000)
\psline[linewidth=\AxesLineWidth,linecolor=GridColor](20.000000,40.000000)(23.360000,40.000000)
\psline[linewidth=\AxesLineWidth,linecolor=GridColor](20.000000,45.000000)(23.360000,45.000000)

{ \footnotesize 
\rput[t](50.000000,-0.982178){$50$}
\rput[t](100.000000,-0.982178){$100$}
\rput[t](150.000000,-0.982178){$150$}
\rput[t](200.000000,-0.982178){$200$}
\rput[t](250.000000,-0.982178){$250$}
\rput[t](300.000000,-0.982178){$300$}
\rput[r](16.640000,0.000000){$0$}
\rput[r](16.640000,5.000000){$5$}
\rput[r](16.640000,10.000000){$10$}
\rput[r](16.640000,15.000000){$15$}
\rput[r](16.640000,20.000000){$20$}
\rput[r](16.640000,25.000000){$25$}
\rput[r](16.640000,30.000000){$30$}
\rput[r](16.640000,35.000000){$35$}
\rput[r](16.640000,40.000000){$40$}
\rput[r](16.640000,45.000000){$45$}
} 

\psframe[linewidth=\AxesLineWidth,dimen=middle](20.000000,0.000000)(300.000000,46.000000)

{ \small 
\rput[b](160.000000,-7.169811){
\begin{tabular}{c}
$R_1$\\
\end{tabular}
}

\rput[t]{90}(-8.387097,20.000000){
\begin{tabular}{c}
$r_{1i}$\\
\end{tabular}
}
} 

\newrgbcolor{color1791.0042}{0  0  1}
\psline[plotstyle=line,linejoin=1,showpoints=false,dotstyle=*,dotsize=\MarkerSize,linestyle=solid,linewidth=\LineWidth,linecolor=color1791.0042]
(300.000000,11.267571)(300.000000,11.267571)
\psline[plotstyle=line,linejoin=1,showpoints=true,dotstyle=*,dotsize=\MarkerSize,linestyle=solid,linewidth=\LineWidth,linecolor=color1791.0042]
(20.000000,0.000000)(30.000000,2.147368)(40.000000,2.676944)(50.000000,2.139593)(60.000000,6.190338)
(70.000000,6.155558)(80.000000,6.482891)(90.000000,8.637891)(100.000000,8.087764)(110.000000,8.495836)
(120.000000,8.832809)(130.000000,8.930901)(140.000000,9.366922)(150.000000,9.926121)(160.000000,10.384524)
(170.000000,10.666865)(180.000000,10.845026)(190.000000,10.975993)(200.000000,11.046529)(210.000000,11.079651)
(220.000000,11.108911)(230.000000,11.135279)(240.000000,11.159120)(250.000000,11.180863)(260.000000,11.200836)
(270.000000,11.219273)(280.000000,11.236368)(290.000000,11.252635)(300.000000,11.267571)

\newrgbcolor{color1792.0037}{0         0.5           0}
\psline[plotstyle=line,linejoin=1,showpoints=false,dotstyle=Bsquare,dotsize=\MarkerSize,linestyle=solid,linewidth=\LineWidth,linecolor=color1792.0037]
(300.000000,21.941951)(300.000000,21.941951)
\psline[plotstyle=line,linejoin=1,showpoints=true,dotstyle=Bsquare,dotsize=\MarkerSize,linestyle=solid,linewidth=\LineWidth,linecolor=color1792.0037]
(20.000000,0.000000)(40.000000,0.000000)(50.000000,9.942648)(60.000000,12.365699)(70.000000,12.447476)
(80.000000,12.989398)(90.000000,17.181115)(100.000000,16.525012)(110.000000,17.314609)(120.000000,17.635547)
(130.000000,17.871769)(140.000000,18.664396)(150.000000,19.637127)(160.000000,20.435227)(170.000000,20.923228)
(180.000000,21.231801)(190.000000,21.454332)(200.000000,21.572751)(210.000000,21.628136)(220.000000,21.677040)
(230.000000,21.721094)(240.000000,21.760915)(250.000000,21.797223)(260.000000,21.830569)(270.000000,21.861345)
(280.000000,21.889878)(290.000000,21.917027)(300.000000,21.941951)

\newrgbcolor{color1793.0037}{1  0  0}
\psline[plotstyle=line,linejoin=1,showpoints=false,dotstyle=Bo,dotsize=\MarkerSize,linestyle=solid,linewidth=\LineWidth,linecolor=color1793.0037]
(300.000000,34.721322)(300.000000,34.721322)
\psline[plotstyle=line,linejoin=1,showpoints=true,dotstyle=Bo,dotsize=\MarkerSize,linestyle=solid,linewidth=\LineWidth,linecolor=color1793.0037]
(20.000000,0.000000)(70.000000,0.000000)(80.000000,0.130594)(90.000000,0.415519)(100.000000,7.149753)
(110.000000,11.529334)(120.000000,13.453544)(130.000000,17.887719)(140.000000,22.381571)(150.000000,26.462350)
(160.000000,29.491088)(170.000000,31.330654)(180.000000,32.491730)(190.000000,33.232825)(200.000000,33.601617)
(210.000000,33.770554)(220.000000,33.919359)(230.000000,34.053159)(240.000000,34.173924)(250.000000,34.283905)
(260.000000,34.384820)(270.000000,34.477881)(280.000000,34.564099)(290.000000,34.646089)(300.000000,34.721322)

\newrgbcolor{color1794.0037}{0        0.75        0.75}
\psline[plotstyle=line,linejoin=1,showpoints=false,dotstyle=Bpentagon,dotsize=\MarkerSize,linestyle=solid,linewidth=\LineWidth,linecolor=color1794.0037]
(300.000000,19.820062)(300.000000,19.820062)
\psline[plotstyle=line,linejoin=1,showpoints=true,dotstyle=Bpentagon,dotsize=\MarkerSize,linestyle=solid,linewidth=\LineWidth,linecolor=color1794.0037]
(20.000000,0.000000)(70.000000,0.000000)(80.000000,0.014995)(90.000000,0.110033)(100.000000,0.281993)
(110.000000,0.456036)(120.000000,1.491031)(130.000000,1.929286)(140.000000,2.079900)(150.000000,2.514073)
(160.000000,3.055001)(170.000000,3.919724)(180.000000,5.131704)(190.000000,6.420987)(200.000000,7.822190)
(210.000000,8.976774)(220.000000,10.142902)(230.000000,11.327491)(240.000000,12.521140)(250.000000,13.722528)
(260.000000,14.930571)(270.000000,16.142621)(280.000000,17.356709)(290.000000,18.599094)(300.000000,19.820062)

\newrgbcolor{color1795.0035}{0.75           0        0.75}
\psline[plotstyle=line,linejoin=1,showpoints=false,dotstyle=Basterisk,dotsize=\MarkerSize,linestyle=solid,linewidth=\LineWidth,linecolor=color1795.0035]
(300.000000,25.674544)(300.000000,25.674544)
\psline[plotstyle=line,linejoin=1,showpoints=true,dotstyle=Basterisk,dotsize=\MarkerSize,linestyle=solid,linewidth=\LineWidth,linecolor=color1795.0035]
(20.000000,0.000000)(70.000000,0.000000)(80.000000,0.021952)(90.000000,0.158872)(100.000000,0.414949)
(110.000000,0.679019)(120.000000,2.127807)(130.000000,2.725629)(140.000000,2.973668)(150.000000,3.366219)
(160.000000,4.056499)(170.000000,5.282275)(180.000000,6.921985)(190.000000,8.652108)(200.000000,10.487510)
(210.000000,11.971680)(220.000000,13.463892)(230.000000,14.973746)(240.000000,16.489846)(250.000000,18.011047)
(260.000000,19.536417)(270.000000,21.062992)(280.000000,22.588628)(290.000000,24.146510)(300.000000,25.674544)

\newrgbcolor{color1796.0035}{0.75        0.75           0}
\psline[plotstyle=line,linejoin=1,showpoints=false,dotstyle=B|,dotsize=\MarkerSize,linestyle=solid,linewidth=\LineWidth,linecolor=color1796.0035]
(300.000000,36.361580)(300.000000,36.361580)
\psline[plotstyle=line,linejoin=1,showpoints=true,dotstyle=B|,dotsize=\MarkerSize,linestyle=solid,linewidth=\LineWidth,linecolor=color1796.0035]
(20.000000,0.000000)(70.000000,0.000000)(80.000000,0.033476)(90.000000,0.245870)(100.000000,0.519565)
(110.000000,0.936731)(120.000000,3.180306)(130.000000,4.060566)(140.000000,4.511487)(150.000000,4.932335)
(160.000000,5.972770)(170.000000,7.736844)(180.000000,10.171549)(190.000000,12.747392)(200.000000,15.404947)
(210.000000,17.499205)(220.000000,19.591116)(230.000000,21.695656)(240.000000,23.798157)(250.000000,25.898134)
(260.000000,27.995220)(270.000000,30.086134)(280.000000,32.168644)(290.000000,34.288454)(300.000000,36.361580)

{ \small 
\rput[tr](119.280000,45.035645){%
\psshadowbox[framesep=0pt,linewidth=\AxesLineWidth]{\psframebox*{\begin{tabular}{l}
\Rnode{a1}{\hspace*{0.0ex}} \hspace*{0.7cm} \Rnode{a2}{~~$ i = 13$} \\
\Rnode{a3}{\hspace*{0.0ex}} \hspace*{0.7cm} \Rnode{a4}{~~$ i = 14$} \\
\Rnode{a5}{\hspace*{0.0ex}} \hspace*{0.7cm} \Rnode{a6}{~~$ i = 15$} \\
\Rnode{a7}{\hspace*{0.0ex}} \hspace*{0.7cm} \Rnode{a8}{~~$ i = 16$} \\
\Rnode{a9}{\hspace*{0.0ex}} \hspace*{0.7cm} \Rnode{a10}{~~$i = 17$} \\
\Rnode{a11}{\hspace*{0.0ex}} \hspace*{0.7cm} \Rnode{a12}{~~$ i = 18$} \\
\end{tabular}}
\ncline[linestyle=solid,linewidth=\LineWidth,linecolor=color1791.0042]{a1}{a2}
\ncput{\psdot[dotstyle=*,dotsize=\MarkerSize,linecolor=color1791.0042]}
\ncline[linestyle=solid,linewidth=\LineWidth,linecolor=color1792.0037]{a3}{a4}
\ncput{\psdot[dotstyle=Bsquare,dotsize=\MarkerSize,linecolor=color1792.0037]}
\ncline[linestyle=solid,linewidth=\LineWidth,linecolor=color1793.0037]{a5}{a6}
\ncput{\psdot[dotstyle=Bo,dotsize=\MarkerSize,linecolor=color1793.0037]}
\ncline[linestyle=solid,linewidth=\LineWidth,linecolor=color1794.0037]{a7}{a8}
\ncput{\psdot[dotstyle=Bpentagon,dotsize=\MarkerSize,linecolor=color1794.0037]}
\ncline[linestyle=solid,linewidth=\LineWidth,linecolor=color1795.0035]{a9}{a10}
\ncput{\psdot[dotstyle=Basterisk,dotsize=\MarkerSize,linecolor=color1795.0035]}
\ncline[linestyle=solid,linewidth=\LineWidth,linecolor=color1796.0035]{a11}{a12}
\ncput{\psdot[dotstyle=B|,dotsize=\MarkerSize,linecolor=color1796.0035]}

}%
}%
} 

\end{pspicture}%
}
\subfigure[The allocated rates $r_{2i}$ from the $2^{nd}$ carrier eNodeB to the $3^{rd}$ group of users.]{
\label{fig:sim:rates32}
%
\psset{xunit=0.003571\plotwidth,yunit=0.013624\plotwidth}%
\begin{pspicture}(-8.387097,-4.453125)(307.741935,30.937500)%


\psline[linewidth=\AxesLineWidth,linecolor=GridColor](50.000000,0.000000)(50.000000,0.611507)
\psline[linewidth=\AxesLineWidth,linecolor=GridColor](100.000000,0.000000)(100.000000,0.611507)
\psline[linewidth=\AxesLineWidth,linecolor=GridColor](150.000000,0.000000)(150.000000,0.611507)
\psline[linewidth=\AxesLineWidth,linecolor=GridColor](200.000000,0.000000)(200.000000,0.611507)
\psline[linewidth=\AxesLineWidth,linecolor=GridColor](250.000000,0.000000)(250.000000,0.611507)
\psline[linewidth=\AxesLineWidth,linecolor=GridColor](300.000000,0.000000)(300.000000,0.611507)
\psline[linewidth=\AxesLineWidth,linecolor=GridColor](20.000000,0.000000)(23.360000,0.000000)
\psline[linewidth=\AxesLineWidth,linecolor=GridColor](20.000000,5.000000)(23.360000,5.000000)
\psline[linewidth=\AxesLineWidth,linecolor=GridColor](20.000000,10.000000)(23.360000,10.000000)
\psline[linewidth=\AxesLineWidth,linecolor=GridColor](20.000000,15.000000)(23.360000,15.000000)
\psline[linewidth=\AxesLineWidth,linecolor=GridColor](20.000000,20.000000)(23.360000,20.000000)
\psline[linewidth=\AxesLineWidth,linecolor=GridColor](20.000000,25.000000)(23.360000,25.000000)
\psline[linewidth=\AxesLineWidth,linecolor=GridColor](20.000000,30.000000)(23.360000,30.000000)

{ \footnotesize 
\rput[t](50.000000,-0.611507){$50$}
\rput[t](100.000000,-0.611507){$100$}
\rput[t](150.000000,-0.611507){$150$}
\rput[t](200.000000,-0.611507){$200$}
\rput[t](250.000000,-0.611507){$250$}
\rput[t](300.000000,-0.611507){$300$}
\rput[r](16.640000,0.000000){$0$}
\rput[r](16.640000,5.000000){$5$}
\rput[r](16.640000,10.000000){$10$}
\rput[r](16.640000,15.000000){$15$}
\rput[r](16.640000,20.000000){$20$}
\rput[r](16.640000,25.000000){$25$}
\rput[r](16.640000,30.000000){$30$}
} 

\psframe[linewidth=\AxesLineWidth,dimen=middle](20.000000,0.000000)(300.000000,30.000000)

{ \small 
\rput[b](160.000000,-5.453125){
\begin{tabular}{c}
$R_1$\\
\end{tabular}
}

\rput[t]{90}(-8.387097,15.000000){
\begin{tabular}{c}
$r_{2i}$\\
\end{tabular}
}
} 

\newrgbcolor{color1791.0042}{0  0  1}
\psline[plotstyle=line,linejoin=1,showpoints=false,dotstyle=*,dotsize=\MarkerSize,linestyle=solid,linewidth=\LineWidth,linecolor=color1791.0042]
(300.000000,0.000000)(300.000000,0.000000)
\psline[plotstyle=line,linejoin=1,showpoints=true,dotstyle=*,dotsize=\MarkerSize,linestyle=solid,linewidth=\LineWidth,linecolor=color1791.0042]
(20.000000,10.277265)(30.000000,10.277252)(40.000000,8.086813)(50.000000,8.499581)(60.000000,4.110114)
(70.000000,4.137871)(80.000000,3.819327)(90.000000,1.761198)(100.000000,2.397155)(110.000000,2.100866)
(120.000000,1.891610)(130.000000,1.860199)(140.000000,1.478954)(150.000000,0.965489)(160.000000,0.545517)
(170.000000,0.297572)(180.000000,0.150176)(190.000000,0.046347)(200.000000,0.000000)(300.000000,0.000000)

\newrgbcolor{color1792.0037}{0         0.5           0}
\psline[plotstyle=line,linejoin=1,showpoints=false,dotstyle=Bsquare,dotsize=\MarkerSize,linestyle=solid,linewidth=\LineWidth,linecolor=color1792.0037]
(300.000000,0.000000)(300.000000,0.000000)
\psline[plotstyle=line,linejoin=1,showpoints=true,dotstyle=Bsquare,dotsize=\MarkerSize,linestyle=solid,linewidth=\LineWidth,linecolor=color1792.0037]
(20.000000,20.231062)(40.000000,20.231051)(50.000000,17.590763)(60.000000,7.975424)(70.000000,7.786617)
(80.000000,7.276157)(90.000000,3.240136)(100.000000,4.362707)(110.000000,3.765473)(120.000000,3.395550)
(130.000000,3.272174)(140.000000,2.571871)(150.000000,1.670437)(160.000000,0.942435)(170.000000,0.512099)
(180.000000,0.255056)(190.000000,0.077951)(200.000000,0.000000)(300.000000,0.000000)

\newrgbcolor{color1793.0037}{1  0  0}
\psline[plotstyle=line,linejoin=1,showpoints=false,dotstyle=Bo,dotsize=\MarkerSize,linestyle=solid,linewidth=\LineWidth,linecolor=color1793.0037]
(300.000000,0.000000)(300.000000,0.000000)
\psline[plotstyle=line,linejoin=1,showpoints=true,dotstyle=Bo,dotsize=\MarkerSize,linestyle=solid,linewidth=\LineWidth,linecolor=color1793.0037]
(20.000000,17.598795)(30.000000,17.597874)(40.000000,18.693171)(50.000000,19.816013)(60.000000,26.737949)
(70.000000,26.812910)(80.000000,27.140849)(90.000000,29.222771)(100.000000,25.025277)(110.000000,21.029939)
(120.000000,18.461664)(130.000000,14.384617)(140.000000,10.182376)(150.000000,6.309309)(160.000000,3.510692)
(170.000000,1.854074)(180.000000,0.846848)(190.000000,0.245039)(200.000000,0.000000)(300.000000,0.000000)

\newrgbcolor{color1794.0037}{0        0.75        0.75}
\psline[plotstyle=line,linejoin=1,showpoints=false,dotstyle=Bpentagon,dotsize=\MarkerSize,linestyle=solid,linewidth=\LineWidth,linecolor=color1794.0037]
(300.000000,0.000000)(300.000000,0.000000)
\psline[plotstyle=line,linejoin=1,showpoints=true,dotstyle=Bpentagon,dotsize=\MarkerSize,linestyle=solid,linewidth=\LineWidth,linecolor=color1794.0037]
(20.000000,0.430868)(40.000000,0.430861)(50.000000,0.429420)(60.000000,0.443534)(70.000000,0.444517)
(80.000000,0.450794)(90.000000,0.530787)(100.000000,0.583687)(110.000000,0.849805)(120.000000,0.640269)
(130.000000,0.841757)(140.000000,1.371437)(150.000000,1.631446)(160.000000,1.790528)(170.000000,1.656795)
(180.000000,1.195733)(190.000000,0.655834)(200.000000,0.000000)(300.000000,0.000000)

\newrgbcolor{color1795.0035}{0.75           0        0.75}
\psline[plotstyle=line,linejoin=1,showpoints=false,dotstyle=Basterisk,dotsize=\MarkerSize,linestyle=solid,linewidth=\LineWidth,linecolor=color1795.0035]
(300.000000,0.000000)(300.000000,0.000000)
\psline[plotstyle=line,linejoin=1,showpoints=true,dotstyle=Basterisk,dotsize=\MarkerSize,linestyle=solid,linewidth=\LineWidth,linecolor=color1795.0035]
(20.000000,0.619143)(40.000000,0.619132)(50.000000,0.617016)(60.000000,0.637731)(70.000000,0.639171)
(80.000000,0.648391)(90.000000,0.765923)(100.000000,0.832105)(110.000000,1.187597)(120.000000,0.874583)
(130.000000,1.135185)(140.000000,1.801122)(150.000000,2.331690)(160.000000,2.558943)(170.000000,2.293998)
(180.000000,1.631964)(190.000000,0.873114)(200.000000,0.000000)(300.000000,0.000000)

\newrgbcolor{color1796.0035}{0.75        0.75           0}
\psline[plotstyle=line,linejoin=1,showpoints=false,dotstyle=B|,dotsize=\MarkerSize,linestyle=solid,linewidth=\LineWidth,linecolor=color1796.0035]
(300.000000,0.000000)(300.000000,0.000000)
\psline[plotstyle=line,linejoin=1,showpoints=true,dotstyle=B|,dotsize=\MarkerSize,linestyle=solid,linewidth=\LineWidth,linecolor=color1796.0035]
(20.000000,0.843079)(40.000000,0.843062)(50.000000,0.839854)(60.000000,0.871298)(70.000000,0.873487)
(80.000000,0.887467)(90.000000,1.066464)(100.000000,1.292899)(110.000000,1.836118)(120.000000,1.329752)
(130.000000,1.745510)(140.000000,2.635234)(150.000000,3.559970)(160.000000,3.895033)(170.000000,3.524354)
(180.000000,2.480649)(190.000000,1.291363)(200.000000,0.000000)(300.000000,0.000000)

{ \small 
\rput[tr](293.280000,28.776987){%
\psshadowbox[framesep=0pt,linewidth=\AxesLineWidth]{\psframebox*{\begin{tabular}{l}
\Rnode{a1}{\hspace*{0.0ex}} \hspace*{0.7cm} \Rnode{a2}{~~$ i = 13$} \\
\Rnode{a3}{\hspace*{0.0ex}} \hspace*{0.7cm} \Rnode{a4}{~~$ i = 14$} \\
\Rnode{a5}{\hspace*{0.0ex}} \hspace*{0.7cm} \Rnode{a6}{~~$ i = 15$} \\
\Rnode{a7}{\hspace*{0.0ex}} \hspace*{0.7cm} \Rnode{a8}{~~$i = 16$} \\
\Rnode{a9}{\hspace*{0.0ex}} \hspace*{0.7cm} \Rnode{a10}{~~$i = 17$} \\
\Rnode{a11}{\hspace*{0.0ex}} \hspace*{0.7cm} \Rnode{a12}{~~$ i = 18$} \\
\end{tabular}}
\ncline[linestyle=solid,linewidth=\LineWidth,linecolor=color1791.0042]{a1}{a2}
\ncput{\psdot[dotstyle=*,dotsize=\MarkerSize,linecolor=color1791.0042]}
\ncline[linestyle=solid,linewidth=\LineWidth,linecolor=color1792.0037]{a3}{a4}
\ncput{\psdot[dotstyle=Bsquare,dotsize=\MarkerSize,linecolor=color1792.0037]}
\ncline[linestyle=solid,linewidth=\LineWidth,linecolor=color1793.0037]{a5}{a6}
\ncput{\psdot[dotstyle=Bo,dotsize=\MarkerSize,linecolor=color1793.0037]}
\ncline[linestyle=solid,linewidth=\LineWidth,linecolor=color1794.0037]{a7}{a8}
\ncput{\psdot[dotstyle=Bpentagon,dotsize=\MarkerSize,linecolor=color1794.0037]}
\ncline[linestyle=solid,linewidth=\LineWidth,linecolor=color1795.0035]{a9}{a10}
\ncput{\psdot[dotstyle=Basterisk,dotsize=\MarkerSize,linecolor=color1795.0035]}
\ncline[linestyle=solid,linewidth=\LineWidth,linecolor=color1796.0035]{a11}{a12}
\ncput{\psdot[dotstyle=B|,dotsize=\MarkerSize,linecolor=color1796.0035]}

}%
}%
} 

\end{pspicture}%
}
\caption{The allocated rates $r_{li}$ from the $l^{th}$ carrier eNodeB to the $3^{rd}$ group of users with $1^{st}$ carrier eNodeB rate $20<R_1<300$ and $2^{nd}$ carrier eNodeB rate fixed at $R_2=100$.}
\label{fig:sim:rates_3rd_group}
\end{figure}

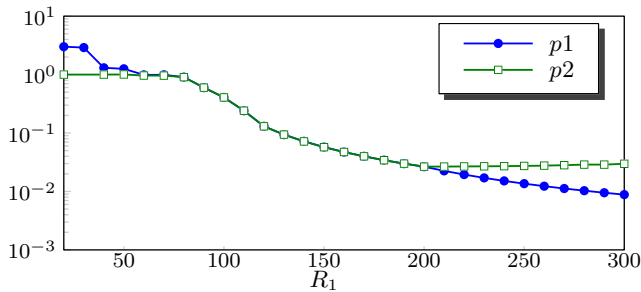
\begin{figure}
%
\psset{xunit=0.003333\plotwidth,yunit=0.097177\plotwidth}%
\begin{pspicture}(-15.898618,-3.444444)(308.294931,1.093567)%


\psline[linewidth=\AxesLineWidth,linecolor=GridColor](20.000000,-3.000000)(20.000000,-2.939141)
\psline[linewidth=\AxesLineWidth,linecolor=GridColor](50.000000,-3.000000)(50.000000,-2.939141)
\psline[linewidth=\AxesLineWidth,linecolor=GridColor](100.000000,-3.000000)(100.000000,-2.939141)
\psline[linewidth=\AxesLineWidth,linecolor=GridColor](150.000000,-3.000000)(150.000000,-2.939141)
\psline[linewidth=\AxesLineWidth,linecolor=GridColor](200.000000,-3.000000)(200.000000,-2.939141)
\psline[linewidth=\AxesLineWidth,linecolor=GridColor](250.000000,-3.000000)(250.000000,-2.939141)
\psline[linewidth=\AxesLineWidth,linecolor=GridColor](300.000000,-3.000000)(300.000000,-2.939141)
\psline[linewidth=\AxesLineWidth,linecolor=GridColor](20.000000,-3.000000)(23.600000,-3.000000)
\psline[linewidth=\AxesLineWidth,linecolor=GridColor](20.000000,-2.000000)(23.600000,-2.000000)
\psline[linewidth=\AxesLineWidth,linecolor=GridColor](20.000000,-1.000000)(23.600000,-1.000000)
\psline[linewidth=\AxesLineWidth,linecolor=GridColor](20.000000,0.000000)(23.600000,0.000000)
\psline[linewidth=\AxesLineWidth,linecolor=GridColor](20.000000,1.000000)(23.600000,1.000000)

\psline[linewidth=\AxesLineWidth,linecolor=GridColor](20.000000,-2.698970)(22.400000,-2.698970)
\psline[linewidth=\AxesLineWidth,linecolor=GridColor](20.000000,-2.522879)(22.400000,-2.522879)
\psline[linewidth=\AxesLineWidth,linecolor=GridColor](20.000000,-2.397940)(22.400000,-2.397940)
\psline[linewidth=\AxesLineWidth,linecolor=GridColor](20.000000,-2.301030)(22.400000,-2.301030)
\psline[linewidth=\AxesLineWidth,linecolor=GridColor](20.000000,-2.221849)(22.400000,-2.221849)
\psline[linewidth=\AxesLineWidth,linecolor=GridColor](20.000000,-2.154902)(22.400000,-2.154902)
\psline[linewidth=\AxesLineWidth,linecolor=GridColor](20.000000,-2.096910)(22.400000,-2.096910)
\psline[linewidth=\AxesLineWidth,linecolor=GridColor](20.000000,-2.045757)(22.400000,-2.045757)
\psline[linewidth=\AxesLineWidth,linecolor=GridColor](20.000000,-1.698970)(22.400000,-1.698970)
\psline[linewidth=\AxesLineWidth,linecolor=GridColor](20.000000,-1.522879)(22.400000,-1.522879)
\psline[linewidth=\AxesLineWidth,linecolor=GridColor](20.000000,-1.397940)(22.400000,-1.397940)
\psline[linewidth=\AxesLineWidth,linecolor=GridColor](20.000000,-1.301030)(22.400000,-1.301030)
\psline[linewidth=\AxesLineWidth,linecolor=GridColor](20.000000,-1.221849)(22.400000,-1.221849)
\psline[linewidth=\AxesLineWidth,linecolor=GridColor](20.000000,-1.154902)(22.400000,-1.154902)
\psline[linewidth=\AxesLineWidth,linecolor=GridColor](20.000000,-1.096910)(22.400000,-1.096910)
\psline[linewidth=\AxesLineWidth,linecolor=GridColor](20.000000,-1.045757)(22.400000,-1.045757)
\psline[linewidth=\AxesLineWidth,linecolor=GridColor](20.000000,-0.698970)(22.400000,-0.698970)
\psline[linewidth=\AxesLineWidth,linecolor=GridColor](20.000000,-0.522879)(22.400000,-0.522879)
\psline[linewidth=\AxesLineWidth,linecolor=GridColor](20.000000,-0.397940)(22.400000,-0.397940)
\psline[linewidth=\AxesLineWidth,linecolor=GridColor](20.000000,-0.301030)(22.400000,-0.301030)
\psline[linewidth=\AxesLineWidth,linecolor=GridColor](20.000000,-0.221849)(22.400000,-0.221849)
\psline[linewidth=\AxesLineWidth,linecolor=GridColor](20.000000,-0.154902)(22.400000,-0.154902)
\psline[linewidth=\AxesLineWidth,linecolor=GridColor](20.000000,-0.096910)(22.400000,-0.096910)
\psline[linewidth=\AxesLineWidth,linecolor=GridColor](20.000000,-0.045757)(22.400000,-0.045757)
\psline[linewidth=\AxesLineWidth,linecolor=GridColor](20.000000,0.301030)(22.400000,0.301030)
\psline[linewidth=\AxesLineWidth,linecolor=GridColor](20.000000,0.477121)(22.400000,0.477121)
\psline[linewidth=\AxesLineWidth,linecolor=GridColor](20.000000,0.602060)(22.400000,0.602060)
\psline[linewidth=\AxesLineWidth,linecolor=GridColor](20.000000,0.698970)(22.400000,0.698970)
\psline[linewidth=\AxesLineWidth,linecolor=GridColor](20.000000,0.778151)(22.400000,0.778151)
\psline[linewidth=\AxesLineWidth,linecolor=GridColor](20.000000,0.845098)(22.400000,0.845098)
\psline[linewidth=\AxesLineWidth,linecolor=GridColor](20.000000,0.903090)(22.400000,0.903090)
\psline[linewidth=\AxesLineWidth,linecolor=GridColor](20.000000,0.954243)(22.400000,0.954243)

{ \footnotesize 
\rput[t](50.000000,-3.060859){$50$}
\rput[t](100.000000,-3.060859){$100$}
\rput[t](150.000000,-3.060859){$150$}
\rput[t](200.000000,-3.060859){$200$}
\rput[t](250.000000,-3.060859){$250$}
\rput[t](300.000000,-3.060859){$300$}
\rput[r](16.400000,-3.000000){$10^{-3}$}
\rput[r](16.400000,-2.000000){$10^{-2}$}
\rput[r](16.400000,-1.000000){$10^{-1}$}
\rput[r](16.400000,0.000000){$10^{0}$}
\rput[r](16.400000,1.000000){$10^{1}$}
} 

\psframe[linewidth=\AxesLineWidth,dimen=middle](20.000000,-3.000000)(300.000000,1.000000)

{ \small 
\rput[b](150.000000,-3.774444){
\begin{tabular}{c}
$R_1$\\
\end{tabular}
}
} 

\newrgbcolor{color1791.0042}{0  0  1}
\psline[plotstyle=line,linejoin=1,showpoints=false,dotstyle=*,dotsize=\MarkerSize,linestyle=solid,linewidth=\LineWidth,linecolor=color1791.0042]
(300.000000,11.267571)(300.000000,11.267571)
\psline[plotstyle=line,linejoin=1,showpoints=true,dotstyle=*,dotsize=\MarkerSize,linestyle=solid,linewidth=\LineWidth,linecolor=color1791.0042]
(20.000000,0.477122)(30.000000,0.462306)(40.000000,0.116102)(50.000000,0.098945)(60.000000,-0.004721)
(70.000000,-0.003903)(80.000000,-0.043913)(90.000000,-0.222558)(100.000000,-0.391209)(110.000000,-0.618330)
(120.000000,-0.884601)(130.000000,-1.027720)(140.000000,-1.142078)(150.000000,-1.240871)(160.000000,-1.326228)
(170.000000,-1.398989)(180.000000,-1.464764)(190.000000,-1.523826)(200.000000,-1.575851)(210.000000,-1.647423)
(220.000000,-1.710694)(230.000000,-1.767743)(240.000000,-1.819346)(250.000000,-1.866424)(260.000000,-1.909683)
(270.000000,-1.949624)(280.000000,-1.986665)(290.000000,-2.021920)(300.000000,-2.054294)

\newrgbcolor{color1792.0037}{0         0.5           0}
\psline[plotstyle=line,linejoin=1,showpoints=false,dotstyle=Bsquare,dotsize=\MarkerSize,linestyle=solid,linewidth=\LineWidth,linecolor=color1792.0037]
(300.000000,21.941951)(300.000000,21.941951)
\psline[plotstyle=line,linejoin=1,showpoints=true,dotstyle=Bsquare,dotsize=\MarkerSize,linestyle=solid,linewidth=\LineWidth,linecolor=color1792.0037]
(20.000000,-0.000011)(40.000000,-0.000001)(50.000000,0.001885)(60.000000,-0.016327)(70.000000,-0.017572)
(80.000000,-0.043886)(90.000000,-0.223000)(100.000000,-0.390852)(110.000000,-0.618195)(120.000000,-0.885541)
(130.000000,-1.027120)(140.000000,-1.144105)(150.000000,-1.242140)(160.000000,-1.324721)(170.000000,-1.398760)
(180.000000,-1.465071)(190.000000,-1.523622)(200.000000,-1.575516)(210.000000,-1.574261)(220.000000,-1.571817)
(230.000000,-1.569832)(240.000000,-1.567069)(250.000000,-1.563227)(260.000000,-1.557896)(270.000000,-1.550520)
(280.000000,-1.540351)(290.000000,-1.540351)(300.000000,-1.526403)

{ \small 
\rput[tl](207.200000,0.878282){%
\psshadowbox[framesep=0pt,linewidth=\AxesLineWidth]{\psframebox*{\begin{tabular}{l}
\Rnode{a1}{\hspace*{0.0ex}} \hspace*{0.7cm} \Rnode{a2}{~~$p1$} \\
\Rnode{a3}{\hspace*{0.0ex}} \hspace*{0.7cm} \Rnode{a4}{~~$p2$} \\
\end{tabular}}
\ncline[linestyle=solid,linewidth=\LineWidth,linecolor=color1791.0042]{a1}{a2}
\ncput{\psdot[dotstyle=*,dotsize=\MarkerSize,linecolor=color1791.0042]}
\ncline[linestyle=solid,linewidth=\LineWidth,linecolor=color1792.0037]{a3}{a4}
\ncput{\psdot[dotstyle=Bsquare,dotsize=\MarkerSize,linecolor=color1792.0037]}

}%
}%
} 

\end{pspicture}%
\caption{The $1^{st}$ carrier shadow price $p_1$ and $2^{nd}$ carrier shadow price $p_2$ with the $1^{st}$ carrier eNodeB rate $20<R_1<300$ and the $2^{nd}$ carrier eNodeB rate $R_2=100$.}
\label{fig:sim:shadow_price}
\end{figure}
\section{Simulation Results}\label{sec:sim}

Algorithm (\ref{alg:UE_FK}) and (\ref{alg:eNodeB_FK}) were applied to various logarithmic and sigmoidal-like utility functions with different parameters in MATLAB. The simulation results showed convergence to the global optimal solution. In this section, we present the simulation results of two carriers and 18 UEs shown in Figure \ref{fig:sim:System_Model}. The UEs are divided into three groups. The $1^{st}$ group is connected to $1^{st}$ carrier eNodeB only (index $i=1,2,3,4,5,6$), the $2^{nd}$ group is connected to $2^{nd}$ carrier eNodeB only (index $i=7,8,9,10,11,12$), and the $3^{rd}$ group is connected to both $1^{st}$ and $2^{nd}$ carriers eNodeBs (index $i=13,14,15,16,17,18$). We use three normalized sigmoidal-like functions that are expressed by equation (\ref{eqn:sigmoid}) with different parameters. The used parameters are $a = 5$, $b=10$ corresponding to a sigmoidal-like function that is an approximation to a step function at rate $r =10$ (e.g. VoIP and is the utility of UEs with indexes $i=1,7,
13$), $a = 3$, $b=20$ corresponding to a sigmoidal-like function that is an approximation of an adaptive real-time application with inflection point at rate $r=20$ (e.g. standard definition video streaming and is the utility of UEs with indexes $i=2,8,14$), and $a = 1$,  $b=30$ corresponding to a sigmoidal-like function that is also an approximation of an adaptive real-time application with inflection point at rate $r=30$ (e.g. high definition video streaming and is the utility of UEs with indexes $i=3,9,15$), as shown in Figure \ref{fig:sim:Utilities}. We use three logarithmic functions that are expressed by equation (\ref{eqn:log}) with $r_{max} =100$ and different $k_i$ parameters which are approximations for delay-tolerant applications (e.g. FTP). We use $k =15$ for UEs with indexes $i=4,10,16$, $k =3$ for UEs with indexes $i=5,11,17$, and $k = 0.5$ for UEs with indexes $i=6,12,18$, as shown in Figure \ref{fig:sim:Utilities}. We set $\delta =10^{-3}$, the $1^{st}$ carrier eNodeB rate $R_1$ takes values 
between 20 and 300 with step of 10, and the $2^{nd}$ carrier eNodeB rate is fixed at $R_2 = 100$. 

In Figure \ref{fig:sim:rates31} and \ref{fig:sim:rates32}, when the resources available at the $2^{nd}$ carrier are more than that at $1^{st}$ carrier, we observe that the resources allocated to the $3^{rd}$ group users are from the $2^{nd}$ carrier. With the increase in the $1^{st}$ carrier eNodeB resources $R_1$, we observe a gradual increase in the $3^{rd}$ group rates allocated from the $1^{st}$ carrier and a gradual decrease from the $2^{nd}$ carrier eNodeB resources. This shift in the resource allocation is due to the decrease in the price of resources from $1^{st}$ carrier as it has more resources. In Figure \ref{fig:sim:shadow_price}, we observe that the shadow price of the $1^{st}$ carrier eNodeB is higher than that of $2^{nd}$ carrier eNodeB for $R_1\leq 50$, approximately equal for $60<R_1\leq200$, and lower for $R_1>200$. Therefore, the $3^{rd}$ group users, that are covered by $1^{st}$ and $2^{nd}$ carrier eNodeBs, receive the minimum price for resources by aggregating resources from both 
eNodeBs as intended by our algorithm.
\section{Conclusion}\label{sec:conclude}
In this paper, we introduced a novel resource allocation optimization problem with joint carrier aggregation. We considered mobile users running both real-time and delay-tolerant applications with utility proportional fairness allocation policy in 4G-LTE system. We proved that the global optimal solution exists and is tractable for mobile stations with logarithmic and sigmoidal-like utility functions. We presented a distributed algorithm for allocating resources from different carriers optimally to mobile users. Our algorithm ensures fairness in the utility percentage achieved by the allocated resources for all users. Therefore, the algorithm gives priority to the users with adaptive real-time applications with guaranteed minimum QoS for all service subscribers. In addition, our algorithm guarantees allocating resources from different carriers with the lowest resource price to mobile users. We showed through simulations that our algorithm converges to the optimal rate allocation with the lowest possible 
resource price.

\bibliographystyle{ieeetr}
\bibliography{pubs}
\end{document}